\documentclass[journal]{IEEEtran}

\usepackage{amsmath,graphicx,amssymb,amsfonts,bm,mathtools,mathrsfs,mathabx}
\usepackage{multicol}
\usepackage{amsfonts}
\usepackage{amssymb}
\usepackage{amsthm}
\usepackage{float}
\usepackage{comment}
\usepackage[linkcolor=mar,colorlinks=true]{hyperref}
\hypersetup{colorlinks,breaklinks,
            linkcolor=[rgb]{0.5,0,0},
            citecolor=[rgb]{0.39,0.7,1.0}}
\usepackage{cleveref}
\usepackage{cite}
\usepackage{enumitem}
\usepackage{outlines}
\usepackage{algorithm,algpseudocode}
\usepackage[skip=0pt]{caption}
\usepackage[skip=0pt]{subcaption}
\usepackage{bm}
\newtheorem{theorem}{Theorem}

\newtheoremstyle{reduced}
  {6pt} 
  {6pt} 
  {} 
  {} 
  {\bfseries} 
  {.} 
  {.5em} 
  {} 
\newtheorem{lemma}[theorem]{Lemma}

\theoremstyle{reduced}
\newtheorem{assumption}{Assumption}

\newtheorem{remark}{Remark}

\newcommand*{\trans}{{\mathsf{T}}}
\newcommand*{\hermtr}{{\mathsf{H}}}

\makeatletter
\renewenvironment{proof}[1][\proofname]{\par
  \vspace{-\topsep}
  \pushQED{\qed}%
  \normalfont
  \topsep3pt \partopsep3pt 
  \trivlist
  \item[\hskip\labelsep
        \itshape
    #1\@addpunct{.}]\ignorespaces
}{%
  \popQED\endtrivlist\@endpefalse
  \addvspace{2pt plus 0pt} 
}
\makeatother

\renewenvironment{thebibliography}[1]{%
   \begin{oldthebibliography}{#1}%
     \setlength{\itemsep}{0ex}%
}%
{%
   \end{oldthebibliography}%
}

\setlist[itemize]{noitemsep, topsep=1pt}
\setlist[itemize]{leftmargin=*}

\begin{document}
\title{Model-Free Learning of Two-Stage Beamformers\\for Passive IRS-Aided Network Design}

%

\author{Hassaan~Hashmi,~\IEEEmembership{Student Member,~IEEE,}
        Spyridon~Pougkakiotis, and~Dionysis~Kalogerias,~\IEEEmembership{Senior Member,~IEEE
        \vspace{-6pt}}
\thanks{H. Hashmi and D. Kalogerias are with the Department of Electrical Engineering, Yale University, New Haven, CT, 06520 USA (email: \{hassaan.hashmi, dionysis.kalogerias\}@yale.edu). 

S. Pougkakiotis is with the School of Science and Engineering, University of Dundee, Scotland, DD1 4HR UK (email: spougkakiotis001@dunde.ac.uk). 

A preliminary version of this work appeared in IEEE ICASSP 2023  \cite{zosga_conf}. This work is supported by a Microsoft gift and by the NSF under grant 2242215.}
}
\maketitle

\begin{abstract}
Electronically tunable metasurfaces, or Intelligent Reflecting Surfaces (IRSs), are a popular technology for achieving high spectral efficiency in modern wireless systems by shaping channels using a multitude of tunable passive reflecting elements. Capitalizing on key practical limitations of IRS-aided beamforming pertaining to system modeling and channel sensing/estimation, we propose a novel, fully data-driven Zeroth-order Stochastic Gradient Ascent (ZoSGA) algorithm for general two-stage (i.e., short/long-term), fully-passive IRS-aided stochastic utility maximization. 
ZoSGA learns long-term optimal IRS beamformers jointly with short-term optimal precoders (e.g., WMMSE-based) via minimal zeroth-order reinforcement and in a strictly model-free fashion, relying solely on the \textit{effective} compound channels observed at the terminals, while being independent of channel models or network/IRS configurations. Another remarkable feature of ZoSGA is being amenable to analysis, enabling us to establish a state-of-the-art (SOTA) convergence rate of the order of $\mathcal{O}(\sqrt{S}\epsilon^{-4})$ under minimal assumptions, where $S$ is the total number of IRS elements, and $\epsilon$ is a desired suboptimality target.
Our numerical results on a standard MISO downlink IRS-aided sumrate maximization setting establish SOTA empirical behavior of ZoSGA as well, consistently and substantially outperforming standard fully model-based baselines. Lastly, we demonstrate that ZoSGA can in fact operate \textit{in the field}, by directly optimizing the capacitances of a varactor-based electromagnetic IRS model (unknown to ZoSGA) on a multiple user/IRS, link-dense network setting, with essentially no computational overheads or performance degradation.
\end{abstract}
\vspace{-4pt}

\vspace{-18pt}
\section{Introduction}\label{sec: intro}

\par The radical growth in the number of mobile and numerous other wireless devices, in particular those with communication modalities requiring high bandwidth and low latency connectivity such as virtual and augmented reality, tactile internet,  internet of things, industrial automation, etc., have pushed existing wireless communication systems to their performance limits. The forthcoming era requires seamless wireless connectivity, necessitating proactive research beyond 5G communications \cite{6g:guo2021enabling,6g:jiang2021road,6g:tataria20216g,6g:you2021towards}. 
\par Recently, 5G-enabling technologies such as mmWave communication, massive multi-input multi-output (MIMO) and dense networks have been vigorously investigated \cite{5g_enable:larsson2014massive, 5g_enable:gotsis2016ultradense, 5g_enable:andrews2014will,5g_enable:shafi20175g}. Still, deployment of both massive MIMO and dense networks incurs high installation/maintenance costs and energy consumption, while mmWaves exhibit physical limitations such as susceptibility to blockages and high propagation losses.
To compensate for such propagation losses, which are a characteristic of higher carrier frequencies, densely packed and highly directional mmWave antennas have been proposed \cite{5g:mmWave_fundamentals,5g:wang2018millimeter}. High directionality, combined with reduced scattering, attenuates mmWave signals in the non-line-of-sight (non-LOS) paths, thus blocking the signals. The reason lies in the physics of signal propagation: mmWaves exhibit a more prismatic propagation, i.e., they diffract less than microwave signals around obstacles \cite{5g:mmWave_challenges_opportunities}.

\begin{figure}[!t]
  \centering
  \centerline{\includegraphics[width=3.19
  in]{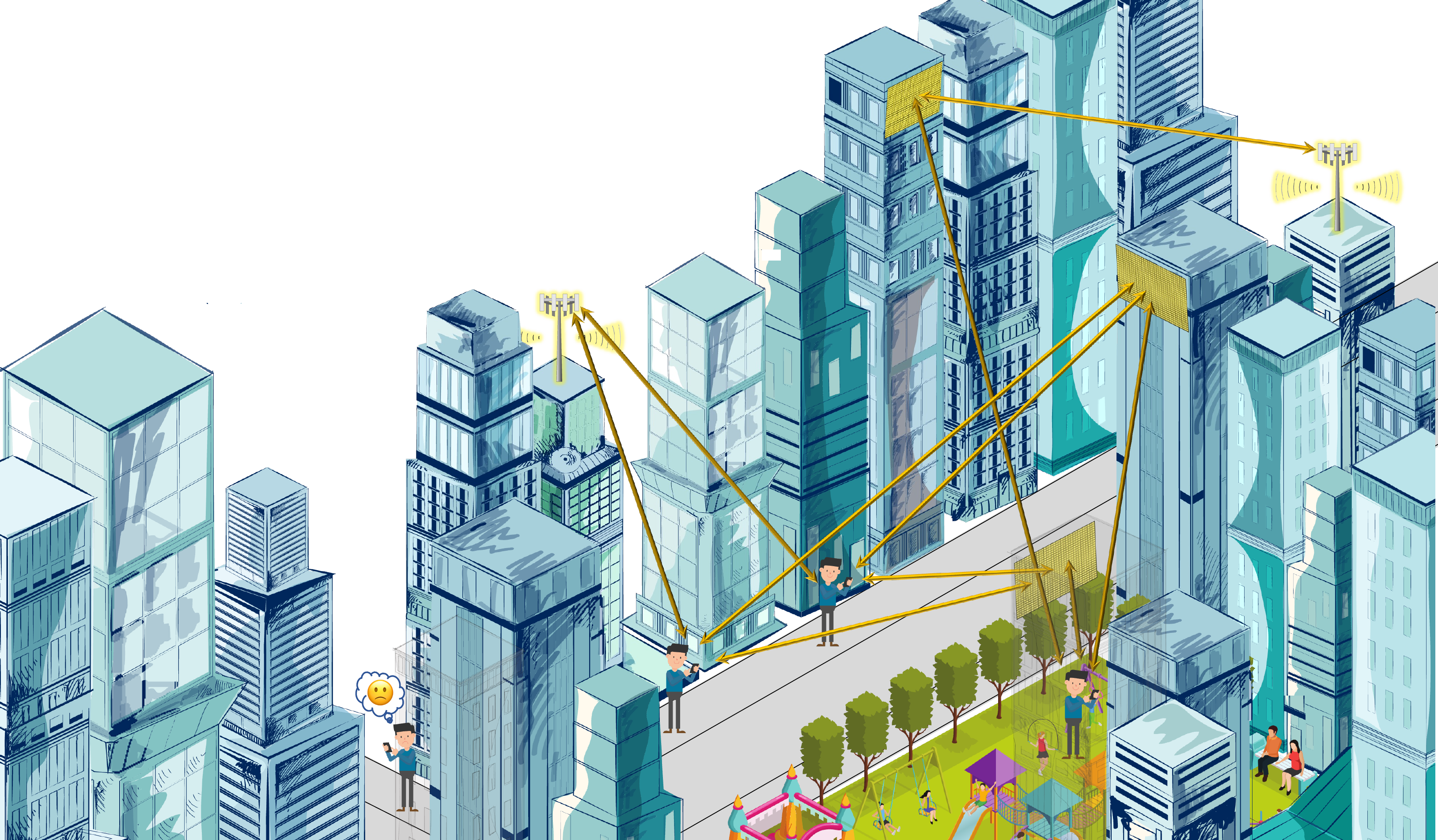}}
\vspace{4bp}
\caption{Concept of an IRS-aided Wireless Network.}
\label{fig:sketch}
\end{figure}

\par Conventional beamforming, which is an established technique for improving Quality-of-Service (QoS) in wireless communications, cannot fully compensate for such non-LOS losses. One brute force solution would be to deploy ultra dense networks, i.e., networks with a large number of small cells, with service ranges of tens to hundreds of meters, allowing higher frequency reuse rates \cite{5g:udn_kamel2016ultra, 5g:udn_valenzuela2018ultra}. With such low ranges, however, ultra-dense network deployment does not demonstrate economy of scale. Indeed, the consumption of energy increases sharply with the number of base stations. Additionally, such dense deployments would also exhibit acute signal interference patterns. 
\par While the aforementioned approaches aim to improve QoS, they may fail to simultaneously satisfy the data-rate, bandwidth, latency, spectral and energy efficiency requirements of 5G-and-beyond technologies. This has necessitated efforts to develop innovative technologies that could meet such requirements, ideally without requiring extra energy and/or deployment or computational costs. An emerging technology for scalably reducing non-LOS losses while circumventing several limitations related to underlying propagation environments is that of \textit{Intelligent Reflecting Surfaces (IRSs, or RISs)}. An IRS is a metasurface comprised of a planar array of passive reflecting elements with tunable parameters, such as phase-shifts and/or amplitude gains of incident signals \cite{irs_design:sievenpiper2003, irs_design:hum2005, kamoda2011pin_diodes, zhao2013varactor1, araghi2022varactor2, rana2023ris, chepuri2023isac}. A concept network in which users (or terminals) are linked with an Access Point (AP) and enjoy improved QoS from the utilization of IRSs is shown in Figure \ref{fig:sketch}.
\par In IRS-aided communications, the goal is to optimally tune the IRS elements along with other resources (such as AP precoders), so as to optimize a certain system utility. A standard setup is that of a weighted sumrate utility in a multi-user multiple-input single-output (MISO) downlink scenario, as depicted in Figure \ref{fig:sketch}, where the goal is to maximize the total downlink rate of a number of users/terminals actively serviced by an AP, while \textit{passively} aided by one or multiple IRSs \cite{sca:guo2020larsson,sca:zhao2020tts, sca:zhao2021qos, sca:yang2021sca}. For this standard setting, a core objective is to jointly optimize the IRS parameters and AP precoders under certain power constraints, noting that AP precoders are usually continuous-valued, while IRS phase-shifts can be either quantized \cite{kamoda2011pin_diodes}, or continuously varying \cite{zhao2013varactor1,araghi2022varactor2, rana2023ris, chepuri2023isac}.
\par \textit{\textbf{Prior Art on IRS Optimization:}} Optimizing IRS parameters is a challenging task, particularly due to three major bottlenecks. First, a sufficiently accurate channel model will most certainly be unknown, and even if known, it depends heavily on network structure and the surrounding environment. Second, by realistically assuming that IRSs are passive components of a wireless network, it might not be possible to continuously and reliably estimate the channel-state-information (CSI) of all the intermediate channels, e.g., from APs to IRSs, and from IRSs to users/terminals \cite{ch_est:zheng2022survey, issues:munochiveyi2021}. While some recent works \cite{issues:liu2020matrix, issues:guo2022cascadest} have investigated methods for cascaded channel estimation with passive IRSs, they typically rely on specific a priori known channel models. It is thus clear that IRS optimization benefits from being free from the need of knowing such detailed CSI. Third, any reasonable IRS phase-shift tuning approach --which should only require \textit{effective} CSI (i.e., the conventional compound channels observed at the end terminals)--, should target as infrequent IRS parameter updates as possible, at least in long-term operation mode. This is in contrast with earlier approaches that consider fully reactive IRS infrastructure demanding resource-wasteful, perpetual IRS control \cite{icsi:wu2019intelligent, icsi:wu2019towards}. Addressing these challenges constitutes an active area of research. In fact, one or more of these bottlenecks are persistent in most available strategies for IRS optimization in wireless networks (e.g., \cite{issues:munochiveyi2021, issues:faisal2022}).
\par Recently, model-based methods relying on some flavor of stochastic successive convex optimization (SSCO) \cite{sca:hong2015decomposition,sca:scutari2013decomposition,sca:yang2016parallel} for weighted sumrate optimization have gained substantial traction. These methods operate over \textit{two time-scale} customized protocols, where reactive precoding vectors at the AP(s) are optimized on a shorter time-scale and non-reactive (static, "long-term") IRS beamformers are optimized on a longer time-scale \cite[Section II-C]{sca:guo2020larsson},\cite{sca:zhao2020tts, sca:yang2021sca, sca:zhao2021qos}. Apart from relying on SSCO, which operates on convex surrogates of the original problem, these model-based approaches require complete knowledge of the network structure and the channel model, along with accurate intermediate CSI (statistic) estimation \cite{ch_est:jian2020modified, ch_est:chen2021low, ch_est:lin2021tensor, ch_est:swindlehurst2022channel}. Such estimates are difficult to obtain because, as an ultimately passive device, an IRS cannot (or should not) transmit and receive pilot signals \cite{ch_est:zheng2022survey}. Consequently, these methods must rely on active sensing at the IRSs, demanding expensive and wasteful IRS implementations. Lastly, any change in the network or channel model incurs a high environment remodeling cost, while the modeling complexity increases dramatically with an increase in the number of intermediate channels.
\par To avoid such limitations, researchers have explored machine learning (ML) methods to optimize IRS-aided networks. As a parallel to channel estimation by hand, offline ML methods have been explored to approximate true CSI models from labelled datasets using function approximators (FAs) \cite{offline:taha2019, offline:balevi2020, offline:yang2021, offline:zhang2021}. These offline approaches are brittle in that the learned models are unable to adapt to slight changes in the network/channel behavior, induced either by movement of the users/terminals, or due to potential environmental factors. Deep Reinforcement Learning (DRL) methods, on the other hand, are adaptive policy learning methods which have been used for joint beamforming optimization. In particular, deep Q-learning based methods have been explored, assuming quantized IRS phase-shifts \cite{dqn:mismar2019, dqn:taha2020, dqn:lee2020}. For continuous phase-shifts, off-policy policy gradient methods have also been explored \cite{ddpg:huang2020, ddpg:yang2020, ddpg:xie2021, ddpg:evmorfos2022}, primarily based on the deep deterministic policy gradient (DDPG) algorithm \cite{ddpg:lillicrap2015}. DRL methods are designed to be \textit{end-to-end}. As such, all intermediate steps such as CSI estimation, channel modeling and beamforming optimization can be \textit{offloaded} to the FAs, the latter approximating state value functions to jointly model the entire optimization task. 
\par Nonetheless, without expert domain knowledge, defining those FAs --e.g., deep neural networks (DNNs)-- often increases the problem complexity; most frequently, FAs are considered as black-box data-driven models, resulting in non-interpretability and lack of robustness. On the other hand, explicit utilization of domain knowledge within the context of FAs often results in overfitting, thus limiting the versatility and transferability of DRL models to distinct environments. Further, such learning-based methods incur increased power/resource consumption as they primarily consider reactive IRS operation, where each model output (phase-shift element values) directly depends on the observed CSI, hence resulting in perpetual IRS control. Combinations of the mixed-timescale iterative approach with ML via deep unfolding models \cite{unfold:liu2021} also suffer from similar limitations.
\par \textit{\textbf{Contributions:}} We develop a \textit{Zeroth-order Stochastic Gradient Ascent (ZoSGA)} algorithm for tackling \textit{fully passive} IRS-assisted utility maximization in a wireless communication setting. We consider two-stage stochastic programming formulations of the problem, in which the first-stage problem consists of an on-average (long-term) optimization of IRS-parameters, while the second-stage problem seeks for optimal instantaneous (short-term) beamformers (e.g., of an AP) associated with a given network instance (occurring every time new --random-- CSI is revealed). ZoSGA tacitly exploits WMMSE \cite{wmmseShi2011} as a standard method for solving the (deterministic) second-stage problems, while remaining agnostic to the channel dynamics or network topology, thus operating in a completely model-free manner. The algorithm relies on minimal system probing and terminal-end effective CSI, both conventionally available (even approximately) regardless of the number or spatial configuration of the IRSs, in sharp contrast to model-based approaches \cite{sca:guo2020larsson, sca:zhao2020tts, sca:yang2021sca, sca:zhao2021qos}. 
\par ZoSGA does not rely on function approximations (unlike DRL methods  \cite{dqn:mismar2019, dqn:taha2020, dqn:lee2020, ddpg:evmorfos2022, ddpg:huang2020, ddpg:yang2020, ddpg:xie2021}), and can be run in real-time since, at each time step, it only requires to probe the network twice (to obtain a sample zeroth-order gradient), greatly improving upon SSCO-based methods, such as TTS-SSCO \cite{sca:zhao2020tts}, which utilize internal (model-based) sampling for approximating stochastic gradients, and --prone to error-- strongly convex surrogate utilities (despite the inherent nonconvexity of the associated optimization problem). Assuming continuous-valued IRS phase-shifts (which are feasible in practice \cite{zhao2013varactor1, araghi2022varactor2, rana2023ris, chepuri2023isac}), ZoSGA treats IRSs as fully passive tunable network elements (without sensing capabilities, extra hardware, or special scheduling requirements), and can be readily applied to a wide range of different wireless network settings. Being model-free, ZoSGA avoids the inherently nonconvex unit modulus constraints associated with the IRS phase-shift elements, by utilizing the (unknown) polar representation of the effective channels. This has far reaching benefits, as we shall see in Section \ref{subsec:Varactor IRSs}, where we simulate \textit{physical} IRS models.
\par After discussing the problem of interest, providing our assumptions and some preliminary technical results (Section \ref{sec: Preliminaries}), we develop and analyze ZoSGA (Section \ref{sec: conv anal}) under a set of general assumptions (and regularity conditions), and we establish a state-of-the-art (SOTA) convergence rate of the order of $\mathcal{O}\left(\sqrt{S}\epsilon^{-4}\right)$, where $S$ is the total number of IRS tunable parameters and $\epsilon$ is a desirable sub-optimality target. We note that our theoretical analysis is novel and involves a minimal set of assumptions, covering a wide range of realistic settings and shedding light into the practical behavior of ZoSGA, fully characterizing its convergence. Unlike most alternative approaches, we do not consider convex approximations of the associated two-stage stochastic problem, and deal with the inherent nonconvexity of the problem by utilizing key results from optimization theory. At the same time, we bypass (without utilizing any approximation) the nonconvexity associated with unit modulus constraints, obtaining a model-free method operating in a standard Euclidean setting. 
\par Specializing to the case of sumrate maximization on a MISO downlink scenario (Section \ref{sec: case study}), we show that most of the technical assumptions imposed by the theory are automatically satisfied, except for some mild regularity conditions on the channel which cannot be avoided. We then numerically demonstrate (Section \ref{sec: Simulations}) that ZoSGA exhibits SOTA performance in a wide-range of scenarios, substantially outperforming the two-time scale method recently proposed in \cite{sca:zhao2020tts}, which is a standard \emph{model-based} SSCO-type SOTA baseline for the problem under consideration. Despite the latter method assuming full knowledge of channel models and spatial network configurations, the model-agnostic ZoSGA reliably learns near-optimal solutions yielding significantly better QoS. Importantly, we also demonstrate the applicability of ZoSGA in a realistic \textit{physical} IRS setting, in which the algorithm has only \textit{indirect access} to the amplitudes and phases of the IRS elements, by tuning the capacitances of certain varactor diodes controlling each IRS element; such IRS tuning \textit{in the field} is a particularly unique feature of ZoSGA, not enjoyed by any other IRS tuning approach in the literature.  Fully reproducible source code of simulation results can be found  \href{https://github.com/hassaanhashmi/zosga}{here}.
\paragraph*{Notation} Let $\|\cdot\|$ denote the induced norm of an associated inner-product space, defined as $\|\bm{x}\| \coloneqq \sqrt{\langle \bm{x}, \bm{x} \rangle}$ for any $\bm{x} \in \mathbb{F}^n$, where $\mathbb{F}$ is a field (assuming that $\mathbb{F} = \mathbb{R}$ or $\mathbb{F} = \mathbb{C}$). In case of a complex vector we use the Hermitian inner product. In case of a matrix we assume that the induced norm is utilized. We assume a complete base probability space $(\Xi,\mathscr{F},P)$, and use ``a.e." to denote ``almost every(where)". For $p \in [1,\infty)$, we use $\mathcal{Z}_p \equiv \mathcal{L}_p(\Xi,\mathscr{F},P;\mathbb{R})$ to denote the space of all $\mathscr{F}$-measurable functions $\phi: \Xi \rightarrow \mathbb{R}$, such that $\int_{\Xi} \lvert \phi \rvert^p dP < \infty$. Given $f \colon \mathbb{R}^n \rightarrow \mathbb{R}$ and $\rho > 0$, we say that $f$ is $\rho$-weakly convex (resp. $\rho$-weakly concave) if $f(\cdot) + \frac{\rho}{2}\|\cdot\|^2$ (resp. $-f(\cdot) + \frac{\rho}{2}\|\cdot\|^2$) is convex.

\vspace{-4pt}
\section{Problem Formulation and Preliminary Results} \label{sec: Preliminaries}
\subsection{Problem Formulation}
\label{subsec: two-stage program}
In an IRS-aided network, the effective channels observed at the terminals of the communication task are treated as functions of both the intermediate channels as well the phase-shift vectors of the IRSs. Thus, we may think of the IRSs as \textit{network-defining} parameters. Indeed, each instance of their phase-shift elements produces a different wireless network altogether. Taking this into account, the objective here is to find an optimal instance of such a wireless network which maximizes a given terminal QoS utility, in accordance with all the underlying physical dynamics.
\par Traditionally, we only optimize the precoding vectors (e.g., at an AP) to maximize a given QoS metric function, and these precoding vectors are optimized for a particular wireless network instance. However, in an IRS-aided wireless network, tuning the IRSs essentially changes the network structure and we have to re-optimize the precoders, responding to this new network.
\par We assume dynamic (i.e. reactive) precoders $\bm{W}$, while (realistically) viewing the IRSs as static (i.e. non-reactive during operation) elements with tunable parameters $\bm{\theta}$ (such as in \cite{sca:guo2020larsson, sca:zhao2020tts, sca:zhao2021qos, sca:yang2021sca}), which encode any propagation feature of the IRSs that is learnable. For example, $\bm{\theta}$ can represent amplitudes and/or phases, or any other tunable element of an IRS (e.g., \textit{physical} varactor tunable capacitances, see \cite{zhao2013varactor1,araghi2022varactor2}).
\par Under this setting, beamforming optimization can be formulated as a \textit{stochastic two-stage problem}. The first-stage problem seeks for an optimal (say in expectation) wireless network by tuning the (static) IRSs' parameters $\bm{\theta}$ assuming optimal instantaneous precoders on random effective channels. The second-stage problem seeks those optimal precoders $\bm{W}$ given a (possibly optimal) network instance set by already fixing the IRSs. 
\par We hereafter assume that the IRSs and the precoders maximize the same QoS utility function. This is an intuitive and standard choice in practical applications, which enables the development of very efficient solution methods. Formally, we are interested in two-stage problems of the form
\begin{equation} \label{eqn: two-stage problem} \tag{2SP}
\begin{split}
    \boxed{
    \max_{\boldsymbol{\theta} \in \Theta}  \mathbb{E}\left\{\max_{\boldsymbol{W} \in \mathcal{W}}{F}\left(\bm{W},\boldsymbol{H}(\boldsymbol{\theta},\omega)\right) \right\},}
    \end{split}
\end{equation}
where $\mathcal{W}$ is a (known) compact set of feasible dynamic beamformers $\boldsymbol{W}$, and $\Theta\subset \mathbb{R}^S$, is a (known) convex and compact set of feasible IRS parameter values, where $S$ denotes the number of real-valued parameters of the complex-valued phase-shift elements (e.g., amplitudes and phases). The utility function $F \colon \mathbb{C}^{M_U}\times \mathbb{C}^{M_U} \rightarrow \mathbb{R}$ is a (known) function of the precoding vectors $\boldsymbol{W}$ as well as the (unknown) observed effective channels $\boldsymbol{H} \colon \Theta \times \Omega \rightarrow \mathbb{C}^{M_U}$, which in turn are functions of both the IRS parameters and any intermediate random channels (a ``state of nature"), denoted as $\omega \colon \Xi \rightarrow \Omega$ (the statistics of which are unknown). The random vector $\omega$ represents anything that is unknown about the underlying communication system, such as propagation or (compound) interference patterns, internal channel states, or in general the underlying intermediate communication channels. The (observed) effective channels $\bm{H}(\cdot,\omega)$ are assumed to have unknown dynamics, and we are only allowed to evaluate them at specific IRS parameter instances $\bm{\theta}\in \Theta$.

As we discuss in Section \ref{subsec: technical results}, the resulting stochastic bilevel program assumes a common function for the inner- and outer- level programs, allowing for first-order optimization, without the need of computing any second-order information (e.g. as in \cite{Dempe2009}). 
\subsection{Assumptions}
\par In what follows, we provide certain regularity assumptions on \eqref{eqn: two-stage problem} and subsequently prove certain core technical results, allowing us to derive the proposed optimization scheme.
\par \textit{Second-stage problem:} Given some realization $\omega \in \Omega$, and some $\bm{\theta} \in \Theta$, the  second-stage problem reads
\begin{equation} \label{eqn: second-stage problem} \tag{SSP}
\max_{\bm{W} \in \mathcal{W}} \,\, \left\{G(\bm{W},\bm{\theta},\omega)   \triangleq F\left(\bm{W},\bm{H}(\bm{\theta},\omega)\right)\right\}. 
\end{equation}
\noindent Notice that \eqref{eqn: second-stage problem} is deterministic, since we are required to solve this after the state of nature $\omega$ has been revealed.
\par \textit{First-stage problem:} The first-stage problem, which is equivalently given in \eqref{eqn: two-stage problem}, can be provisionally written as
\begin{equation} \label{eqn: first-stage problem} \tag{FSP}
\max_{\bm{\theta} \in \Theta} \,\, \{ f(\bm{\theta}) \triangleq \mathbb{E}\left\{F\left(\bm{W}^*(\bm{\theta},\omega),\bm{H}(\bm{\theta},\omega)\right) \right\}\} 
\end{equation}
\noindent for some $\bm{W}^*(\bm{\theta},\omega) \in \arg\max_{\bm{W} \in \mathcal{W}} F\left(\bm{W},\bm{H}(\bm{\theta},\omega)\right)$. In what follows, we enforce certain regularity conditions on \eqref{eqn: two-stage problem}. 
\begin{assumption} \label{assumption: two-stage problem}
The following conditions are in effect:
\begin{enumerate}[align=parleft,labelsep=0.75cm]
    \item[\textnormal{\textbf{(A1)}}] The function $F \colon \mathbb{C}^{M_U} \times \mathbb{C}^{M_U} \rightarrow \mathbb{R}$ is twice continuously (real) differentiable;
    \item[\textnormal{\textbf{(A2)}}] The sets $\Theta$ and $\mathcal{W}$ are compact, and $\Theta$ is also convex;
    \item[\textnormal{\textbf{(A3)}}] The function $\bm{H}(\cdot,\omega)$ is $B_H$-uniformly bounded on $\Theta$ and twice continuously differentiable on an open set $\mathcal{U} \supset \Theta$, for a.e. $\omega \in \Omega$. Moreover, there exist numbers $L_{H,0}$, $L_{H,1}$, such that $\bm{H}(\cdot,\omega)$ is $L_{H,0}$-Lipschitz continuous with $L_{H,1}$-Lipschitz gradients on $\Theta$ for a.e. $\omega \in \Omega$;
    \item[\textnormal{\textbf{(A4)}}] There exists a positive function $\widetilde{\rho}(\cdot) \in \mathcal{Z}_1$, such that $\max_{\bm{W} \in \mathcal{W}} F(\bm{W},\bm{H}(\cdot,\omega))$ is $\widetilde{\rho}(\omega)$-weakly concave on $\Theta$;
    \item[\textnormal{\textbf{(A5)}}] $F(\bm{W}^*(\bm{\theta},\cdot),\bm{H}(\bm{\theta},\cdot) )\in \mathcal{Z}_2$ is bounded below for all $\bm{\theta} \in \Theta$ and any  $\bm{W}^*(\bm{\theta},\omega) \in \arg\max_{\bm{W} \in \mathcal{W}} F(\bm{W},\bm{H}(\bm{\theta},\omega))$, and we can draw independent and identically distributed (i.i.d.) samples from the law of $\omega$.
\end{enumerate}
\end{assumption}
\begin{remark} \label{remark: assumption on second-stage problem}
Let us observe that Assumption \textnormal{\ref{assumption: two-stage problem}} is very mild, and is informed by our application, i.e. two-stage beamforming for passive IRS-aided network design. Conditions \textnormal{\textbf{(A1)}}--\textnormal{\textbf{(A2)}} are standard and are most often met in practical settings involving IRS-aided wireless communication systems. In particular, twice-continuous (real) differentiability of $F(\cdot,\cdot)$ is standard and subsumes several utility functions of interest (three popular examples are the weighted sumrate utility, the proportional fairness utility, or the harmonic-rate utility; see \cite{Prop_fairness}). Note that \emph{real differentiability} refers to differentiability of the real and imaginary parts of $F$ (see \textnormal{\cite[Section 3.2]{arXiv:Kreutz-Delgado}}). Furthermore, as already noted, by utilizing the \emph{polar form} of the function $\bm{H}(\cdot,\omega)$ (which we do in this work), $\Theta$ (typically) represents a set of phases and amplitudes which can be chosen to be real, compact and convex, without loss of generality. Finally, compactness of $\mathcal{W}$ is also standard, since it typically reflects constraints relating to the available power of the wireless communication system.
\par Condition \textnormal{\textbf{(A3)}} ensures that the compositional function of interest is well-defined and retains its properties on an open set containing $\Theta$, while the bound and Lipschitz random functions associated with $\bm{H}(\cdot,\omega)$ are uniformly bounded in $\omega$. The latter condition could potentially be relaxed (e.g. assuming bounded variance of random functions $B_{H,0}(\omega)$, $L_{H,0}(\omega),\ L_{H,1}(\omega)$), but is imposed for brevity in exposition.
\par \textnormal{\textbf{(A4)}} is a very general condition since weak concavity subsumes a large class of functions (e.g. all Lipschitz smooth functions or all twice continuously differentiable functions on a compact set are weakly concave; see also \textnormal{\cite[Section 2]{SIAMOpt:Davis}} for some additional examples). In particular, it implies that for a.e. $\omega \in \Omega$, there exists $\widetilde{\rho}(\omega) > 0$ such that $\max_{\bm{W} \in \mathcal{W}} F(\bm{W},\bm{H}(\cdot,\omega)) - \frac{\widetilde{\rho}(\omega)}{2}\|\cdot\|^2$ is concave on $\Theta$. Specialized, albeit technical, conditions ensuring that this holds for the objective in \textnormal{\eqref{eqn: first-stage problem}} will be discussed later in Section \textnormal{\ref{sec: case study}}.
\par Finally, condition \textnormal{\textbf{(A5)}} ensures that $f$ is well-defined and bounded on $\Theta$. Notice that the same would hold under the weaker condition $F(\bm{W}^*(\bm{\theta},\cdot),\bm{H}(\bm{\theta},\cdot) ) \in \mathcal{Z}_1$, however \textnormal{\textbf{(A5)}} is utilized later in Lemma \textnormal{\ref{lemma: bound on sample gradient variance}}.
\end{remark}

\vspace{-13pt}
\subsection{Technical Results} \label{subsec: technical results}
\par Since ${F}(\bm{W},\cdot)$ takes a complex input, we need to devise an appropriate gradient generalization for it. To that end, we utilize the so-called \emph{Wirtinger calculus} (see \cite{arXiv:Kreutz-Delgado}). In the following lemma we derive the full \textit{compositional} gradient of ${F}(\bm{W},\bm{H}(\bm{\theta},\omega))$ by following the developments in \cite[Section 4]{arXiv:Kreutz-Delgado}.
\begin{lemma} \label{lemma: gradient of compositional}
For every $\bm{\theta} \in \Theta$, $\bm{W} \in \mathcal{W}$ and a.e. $\omega\in \Omega$, the gradient of ${F}\left(\bm{W},\bm{H}(\bm{\theta},\omega)\right)$ with respect to $\bm{\theta}$ reads
\begin{equation} 
    \begin{split}
& \nabla_{\bm{\theta}} F\left(\bm{W},\bm{H}(\bm{\theta},\omega)\right)  \\ &\quad =\ 2 \nabla_{\bm{\theta}} \Re\left(\bm{H}(\bm{\theta},\omega)\right)\Re\left(\frac{\partial^{\circ}}{\partial \bm{z}} F(\bm{W},\bm{z}) \bigg\vert_{\bm{z} = \bm{H}(\bm{\theta},\omega)} \right)^\top \\  & \qquad + 2 \nabla_{\bm{\theta}} \Im\left(\bm{H}(\bm{\theta},\omega)\right)\Re\left(j\frac{\partial^{\circ}}{\partial \bm{z}} F(\bm{W},\bm{z}) \bigg\vert_{\bm{z} = \bm{H}(\bm{\theta},\omega)} \right)^\top, 
\end{split} \label{eqn: Wirtinger gradient of the compositional objective}
\end{equation}
\noindent where $\frac{\partial^{\circ}}{\partial \bm{z}}(\cdot)$ is the Wirtinger cogradient operator.
\end{lemma}
\begin{proof}
    For a complete proof, see Appendix \ref{apdx: Wirtinger gradient}.
\end{proof}
We proceed by proving that, for a.e. $\omega \in \Omega$, the Wirtinger cogradient of $F$, evaluated at some $\bm{z} = \bm{H}(\bm{\theta},\omega)$ for any $\bm{\theta} \in \Theta$, is bounded by a positive constant, independent of $\omega$.
\begin{lemma} \label{lemma: boundedness of Wirtinger gradient}
Given Assumption  \textnormal{\ref{assumption: two-stage problem}}, there exists a constant $B_F>0$ such that, for all $(\bm{\theta},\bm{W}) \in \Theta \times \mathcal{W}$, and a.e. $\omega \in \Omega$, 
    \[\left\|\frac{\partial^{\circ}}{\partial \bm{z}} F(\bm{W},\bm{z})\big\vert_{\bm{z} = \bm{H}(\bm{\theta},\omega)}\right\| \leq B_F.\]
\end{lemma}
\begin{proof}
\par From \textbf{(A3)} we obtain that the sets $\bm{H}(\Theta,\omega)$, for $\omega \in \Omega$, are uniformly (in $\omega$) compact. Thus, we can find a compact set $\mathcal{V}$, such that $\bm{H}(\Theta,\omega) \subset \mathcal{V}$ for a.e. $\omega \in \Omega$. From \textbf{(A2)} we also have that $\mathcal{W}$ compact. Since $F$ is continuously (real) differentiable on the compact set $\mathcal{V} \times \mathcal{W}$, it follows from the extreme value theorem applied to $ \frac{\partial}{\partial \Re(\bm{z})} F(\bm{W},\bm{z})$ and to $\frac{\partial}{\partial \Im(\bm{z})} F(\bm{W},\bm{z})$, respectively, with $\bm{z} \in \mathcal{V}$, that there exists some constant $B_F$ such that
\[\left\|\frac{\partial}{\partial \Re(\bm{z})} F(\bm{W},\bm{z})\big\vert_{\bm{z} = \bm{H}(\bm{\theta},\omega)} \right\| \leq B_F,\]
\[ \left\|\frac{\partial}{\partial \Im(\bm{z})} F(\bm{W},\bm{z})\big\vert_{\bm{z} = \bm{H}(\bm{\theta},\omega)}\right\| \leq B_F,\]
\noindent for a.e. $\omega \in \Omega$ and any $(\bm{\theta},\bm{W}) \in \Theta \times \mathcal{W}$. To complete the proof, we simply bound the Wirtinger cogradient directly by its definition (adjusting the constant upper bound).
\end{proof}

\par Next, we note that under Assumption \ref{assumption: two-stage problem}, the second-stage problem (as a function of $\theta$, for a.e. $\omega \in \Omega$) is a \emph{lower-$C^2$} function (see \cite{RockafellarLowerC2} for a precise definition). We formalize the consequences of this fact in the following lemma.
\begin{lemma} \label{lemma: weak convexity of sample two-stage function}
   Let Assumption \textnormal{\ref{assumption: two-stage problem}} hold. For a.e. $\omega \in \Omega$, there exists a constant $\widehat{\rho}(\omega) > 0$ such that the mapping $\theta \mapsto \max_{\bm{W} \in \mathcal{W}} F(\bm{W},\bm{H}(\bm{\theta},\omega))$ is $\widehat{\rho}(\omega)$-weakly convex on $\Theta$.
\end{lemma}
\begin{proof}
   From Conditions \textbf{(A1)}--\textbf{(A3)}, it follows easily that the mapping $\theta \mapsto \max_{\bm{W} \in \mathcal{W}} F(\bm{W},\bm{H}(\bm{\theta},\omega))$ is lower-$C^2$ (see \cite{RockafellarLowerC2}). Then, \cite[Theorem 6]{RockafellarLowerC2} yields local weak convexity of this mapping on $\Theta$. Compactness of $\Theta$ then yields the result.
\end{proof}

\par We are now ready to show that $f$ in \eqref{eqn: first-stage problem} is actually continuously differentiable and weakly concave on $\Theta$.
\begin{theorem} \label{thm: gradient of f}
     Let Assumption \textnormal{\ref{assumption: two-stage problem}} hold. Then, for any $\bm{\theta} \in \Theta$, the function $f$ is well-defined and differentiable, with
     \begin{equation} \label{eqn: gradient of f}
     \begin{split}
        \nabla_{\bm{\theta}} f(\bm{\theta}) =&\ \mathbb{E}\left\{\nabla_{\bm{\theta}} F\left(\bm{W},\bm{H}(\bm{\theta},\omega)\right)\big\vert_{\bm{W} = \bm{W}^*(\bm{\theta},\omega)}\right\},
        \end{split}
     \end{equation}
     \noindent for any $\bm{W}^*(\bm{\theta},\omega) \in \arg\max_{\bm{W} \in \mathcal{W}} F\left(\bm{W},\bm{H}(\bm{\theta},\omega)\right)$. Moreover, if $\widehat{\rho}(\cdot) \in \mathcal{Z}_1$, where $\widehat{\rho}$ is given in Lemma \textnormal{\ref{lemma: weak convexity of sample two-stage function}}, then there exists a positive constant ${\rho} \triangleq \max\left\{\mathbb{E}\{\widetilde{\rho}(\omega)\},\mathbb{E}\{\widehat{\rho}(\omega)\}\right\}$, such that $f$ is $\rho$-weakly concave on $\Theta$.
\end{theorem}
\begin{proof}
\par Condition \textbf{(A5)} ensures that $f$ is well-defined  and finite for any $\bm{\theta} \in \Theta$. On the other hand, by combining Condition \textbf{(A4)} with Lemma \ref{lemma: weak convexity of sample two-stage function}, we obtain that the map $\bm{\theta} \mapsto \max_{\bm{W} \in \mathcal{W}} F(\bm{W},\bm{H}(\bm{\theta},\omega))$ is continuously Fr\'echet differentiable (the proof of this fact can be found in \cite[Corollary 4.9]{Vial_WeakConvexity}). Since this mapping is \emph{subdifferentially regular}, and its (generalized) subdifferential is a singleton, it follows that
\[ \nabla_{\bm{\theta}} \max_{\bm{W} \in \mathcal{W}} F\left(\bm{W},\bm{H}(\bm{\theta},\omega)\right) =  \nabla_{\bm{\theta}}F\left(\bm{W},\bm{H}(\bm{\theta},\omega)\right)\big\vert_{\bm{W} = \bm{W}^*(\bm{\theta},\omega)},\]
\noindent for any $\bm{W}^*(\bm{\theta},\omega) \in \arg\max_{\bm{W} \in \mathcal{W}} F\left(\bm{W},\bm{H}(\bm{\theta},\omega)\right)$ (e.g. see \cite{DanskinMinMax}). Similarly, by utilizing the integrability of $\widetilde{\rho}(\cdot)$ and $\widehat{\rho}(\cdot)$, we obtain that $f$ is also continuously Fr\'echet differentiable, (again invoking \cite[Corollary 4.9]{Vial_WeakConvexity}, by noting that $f$ is both $\rho$-weakly convex and $\rho$-weakly concave, with ${\rho} \triangleq \max\left\{\mathbb{E}\{\widetilde{\rho}(\omega)\},\mathbb{E}\{\widehat{\rho}(\omega)\}\right\}$). Using \cite[Theorem 7.44]{SIAM:Shapiro_etal}, the gradient of $f$ reads
\[ \nabla_{\bm{\theta}} f(\bm{\theta})  =  \mathbb{E}\left\{\nabla_{\bm{\theta}} F\left(\bm{W},\bm{H}(\bm{\theta},\omega)\right)\big\vert_{\bm{W} =\bm{W}^*(\bm{\theta},\omega)}\right\} ,\]
\noindent for any $\bm{W}^*(\bm{\theta},\cdot)\in \arg\max_{\bm{W} \in \mathcal{W}} F(\bm{W},\bm{H}(\bm{\theta},\cdot))$. 
\end{proof}
\begin{remark}
    \par We observe that the integrability of the weak convexity constants $\widehat{\rho}(\omega)$ (given in Lemma \textnormal{\ref{lemma: weak convexity of sample two-stage function}}) is a very mild condition and is almost always met in practice (where one usually has a finite collection of scenarios).  As such, for the rest of this article we make the implicit assumption that this holds, i.e. $\widehat{\rho}(\cdot) \in \mathcal{Z}_1$.
\end{remark}

\subsection{Zeroth-order Gradient Approximation of the Channel}

\par Assuming the lack of availability of any first-order information of $\bm{H}(\cdot,\omega)$ for any $\omega \in \Omega$, we will employ a zeroth-order scheme in order to obtain a gradient estimate of $\nabla_{\bm{\theta}} \bm{H}(\cdot,\omega)$, using which we can solve \eqref{eqn: first-stage problem} via a stochastic projected gradient ascent scheme. The proposed method will be based on gradient estimates arising from a two-point stochastic evaluation of $\bm{H}(\cdot,\omega)$ (similar to, among others, \cite{SIAMOPT:KalogeriasPowellZerothOrder,CompMath:Nesterov_etal,arXiv:Pougk-Kal}). From \textbf{(A3)}, we have
\[\nabla_{\bm{\theta}} \bm{H}(\bm{\theta},\omega) = \nabla_{\bm{\theta}} \Re\left(\bm{H}(\bm{\theta},\omega)\right) + j \nabla_{\bm{\theta}} \Im\left(\bm{H}(\bm{\theta},\omega) \right).\]
\noindent We would like to approximate each of the above parts of the gradient using only function evaluations of $\bm{H}(\cdot,\omega)$. We let $\bm{U} \sim \mathcal{N}\left(\bm{0}, \bm{I}\right)$ be a normal random vector, where $\bm{I}$ is the identity matrix of size $S$. Given a smoothing parameter $\mu > 0$, we consider the following gradient estimate
\begin{align}
\nabla_{\bm{\theta}} \bm{H}_{\mu}(\bm{\theta},\omega)\triangleq &  \frac{1}{2\mu}\mathbb{E}\left\{\left(\bm{H}\left(\bm{\theta} \hspace{-1pt}+\hspace{-1pt} \mu \bm{U},\omega\right) \hspace{-1pt}-\hspace{-1pt} \bm{H}\left(\bm{\theta}-\mu \bm{U},\omega\right)\right)\bm{U}^\top \right\}^\top \nonumber \\  \equiv &\nabla_{\bm{\theta}} \bm{H}^R_{\mu}(\bm{\theta},\omega) + j \nabla_{\bm{\theta}} \bm{H}^I_{\mu}(\bm{\theta},\omega), \label{eqn: ZO gradient approximation}
\end{align}
\noindent where $\nabla_{\bm{\theta}} \bm{H}^R_{\mu}(\bm{\theta},\omega) \triangleq  \Re\left(\nabla_{\bm{\theta}} \bm{H}_{\mu}(\bm{\theta},\omega)\right)$ and $\nabla_{\bm{\theta}} \bm{H}^I_{\mu}(\bm{\theta},\omega) \triangleq  \Im\left(\nabla_{\bm{\theta}} \bm{H}_{\mu}(\bm{\theta},\omega)\right).$ Let us notice that given condition \textbf{(A3)}, there exists an open set $\mathcal{U} \supset \Theta$ such that \eqref{eqn: ZO gradient approximation} is still well-defined. Thus, the gradient approximation is valid, even if $\bm{\theta}$ is a point in the boundary of $\Theta$, assuming that an appropriately small $\mu$ is chosen. Observe that the smaller the value of $\mu$ is, the better the aforementioned zeroth-order approximation is. There is a trade-off between approximation accuracy and numerical stability, but in practice we observe that $\mu$ can be chosen to be quite small. 
\par The assumption that the dynamics of $\bm{H}(\cdot,\omega)$ are unknown has multiple benefits. It allows us to bypass any modelling assumptions about the underlying communication channels, which typically incur modelling errors. At the same time, it enables the evaluation of $\bm{H}(\cdot,\omega)$ using polar coordinates. That is, we can (and we do) assume that the IRS parameters $\bm{\theta}$ represent any real-valued parameters determining the complex-valued phase-shift elements of the IRSs (e.g. amplitutes and phases). This allows us to bypass the typical nonconvex unit-modulus constraints that arise when optimizing over complex IRS phase-shifts.
\par The proposed zeroth-order gradient approximation of \eqref{eqn: Wirtinger gradient of the compositional objective}, based on the zeroth-order approximation given in \eqref{eqn: ZO gradient approximation}, reads
\begin{equation} \label{eqn: ZO gradient sample of the compositional function}
\begin{split}
  & \nabla_{\bm{\theta}}^{\mu} F\left(\bm{W},\bm{H}(\bm{\theta},\omega)\right) \\
  &\quad \triangleq \ 2 \nabla_{\bm{\theta}} \bm{H}_{\mu}^R(\bm{\theta},\omega)\left(\Re\left(\frac{\partial^{\circ}}{\partial \bm{z}} F(\bm{W},\bm{z}) \big\vert_{\bm{z} = \bm{H}(\bm{\theta},\omega)} \right)\right)^\top \\ & \qquad + 2\nabla_{\bm{\theta}}  \bm{H}_{\mu}^I(\bm{\theta},\omega)\left(\Re\left(j\frac{\partial^{\circ}}{\partial \bm{z}} F(\bm{W},\bm{z}) \big\vert_{\bm{z} = \bm{H}(\bm{\theta},\omega)} \right)\right)^\top,
   \end{split}
\end{equation}
\noindent where $\bm{H}_{\mu}^R(\bm{\theta},\omega)$ and $\bm{H}_{\mu}^I(\bm{\theta},\omega)$ are defined in \eqref{eqn: ZO gradient approximation}. It should be noted here that the zeroth-order gradient approximation relies on small perturbations in the IRS parameters $\bm{\theta}$. This implies that ZoSGA is not applicable when optimizing IRSs with discrete phases/states. We identify this limitation of ZoSGA as a direction for future investigation.

\section{Algorithm and Convergence Analysis} \label{sec: conv anal}

\par We now derive a zeroth-order projected stochastic gradient method for the solution of \eqref{eqn: first-stage problem}. To that end, we assume the availability of an oracle solving the deterministic problem \eqref{eqn: second-stage problem}. The method treats the unknown function $\bm{H}(\cdot,\omega)$ as a black-box, utilizing samples of the gradient approximation given in \eqref{eqn: ZO gradient approximation}.
\begin{assumption} \label{assumption: oracle assumption}
Given any $\bm{\theta} \in \Theta$ and for a.e. $\omega \in \Omega$, we have access to an oracle yielding a (measurable) optimal solution to \textnormal{\eqref{eqn: second-stage problem}},  and Assumption \textnormal{\ref{assumption: two-stage problem}} holds.
\end{assumption}
\vspace{-10pt}\subsection{A Zeroth-order Projected Stochastic Gradient Ascent}
\par Let us briefly present the proposed (channel-agnostic) zeroth-order projected stochastic gradient ascent. From condition \textbf{(A5)}, we have available i.i.d. samples $\omega \in \Omega$. Thus, from \eqref{eqn: ZO gradient sample of the compositional function}, at every $(\bm{\theta},\bm{W}) \in \Theta\times\mathcal{W}$, and for a.e. $\omega \in \Omega$, we can utilize the following sample gradient approximation
\begin{equation} \label{eqn: ZO gradient of the compositional function}
\begin{split}
 & \bm{D}_{\mu}(\bm{\theta},\omega,\bm{U};\bm{W}) \equiv \bm{D}_{\mu}(\bm{\theta},\omega,\bm{U})\\
 &\qquad \triangleq \bm{\Delta}_{\mu}^R \left(\Re\left(\frac{\partial^{\circ}}{\partial \bm{z}} F(\bm{W},\bm{z}) \bigg\vert_{\bm{z} = \bm{H}(\bm{\theta},\omega)} \right)\right)^\top \\ & \qquad\quad   + \bm{\Delta}_{\mu}^I \left(\Re\left(j\frac{\partial^{\circ}}{\partial \bm{z}} F(\bm{W},\bm{z}) \bigg\vert_{\bm{z} = \bm{H}(\bm{\theta},\omega)} \right)\right)^\top,
\end{split}
\end{equation}
with
\begin{align}\nonumber
    &\hspace{-8pt}\left(\bm{\Delta}_{\mu}^R, \bm{\Delta}_{\mu}^I \right)  
    \\ 
    &\triangleq\frac{1}{\mu}\Big[\left(\Re\left(\bm{\Delta}_{\mu}(\bm{\theta},\omega,\bm{U})\right) \bm{U}^\top \right)^\top \, \left(\Im\left(\bm{\Delta}_{\mu}(\bm{\theta},\omega,\bm{U})\right) \bm{U}^\top \right)^\top\Big], \nonumber
\end{align}
\noindent where $\mu > 0$ is a smoothing parameter, $\bm{U} \sim \mathcal{N}(\bm{0},\bm{I})$ and $\bm{\Delta}_{\mu}(\bm{\theta},\omega,\bm{U}) \triangleq \bm{H}\left(\bm{\theta} + \mu \bm{U},\omega\right) - \bm{H}\left(\bm{\theta}-\mu \bm{U},\omega\right)$. Note that this is simply a sample from \eqref{eqn: ZO gradient sample of the compositional function} and is obtained by probing the wireless network twice with the perturbed IRS parameters $\bm{\theta} + \mu \bm{U}$ and $\bm{\theta} - \mu \bm{U}$, with the induced overhead due to the associated channel estimation effort, implicitly assumed to be within the coherence time of the channel; note that this a standard assumption in the related literature on learning resource policies in wireless systems; see, e.g., the seminal work \cite{2019eisen_pfo}.

\par At this point, it may be also worth noting that ZoSGA requires exactly three effective channels to be estimated, per operational iteration: one to communicate (on which the optimal short-term precoders $\boldsymbol{W}^*$ are calculated), and two more pertaining to the required function evaluations (system probes) for constructing the sample gradient approximations outlined above. Therefore, any conventional scheme can be  employed for channel estimation, \textit{as if no IRS is present in the system}. The proposed method is summarized in Algorithm \ref{Algorithm: ZoSGA}.
\renewcommand{\thealgorithm}{ZoSGA}

\begin{algorithm}[!t]
\caption{Zeroth-order Stochastic Gradient Ascent}
    \label{Algorithm: ZoSGA}
\begin{algorithmic}
\State \textbf{Input:}  $\bm{\theta}_0 \in \Theta$, $\{\eta_t\}_{t \geq 0} \subset \mathbb{R}_+$, $\mu > 0$, and $T > 0$.
\For {($t = 0,1,2,\ldots, T$)}
\State Sample (i.i.d.) $\omega_t \in \Omega$, $\bm{U}_t \sim \mathcal{N}\left(\bm{0},\bm{I}\right)$.
\State Find $\bm{W}^* \in \arg\max_{\bm{W} \in \mathcal{W}} F\left(\bm{W},\bm{H}(\bm{\theta}_t,\omega_t)\right).$
\State Set $\bm{D}_{\mu}(\bm{\theta}_t,\omega_t,\bm{U}_t) \equiv \bm{D}_{\mu}(\bm{\theta}_t,\omega_t,\bm{U}_t;\bm{W}^*)$ as in \eqref{eqn: ZO gradient of the compositional function}.
\State  $\bm{\theta}_{t+1} = \textbf{proj}_{\Theta}\left(\bm{\theta}_t + \eta_t \bm{D}_{\mu}\left(\bm{\theta}_t,\omega_t,\bm{U}_t\right)\right). $
\EndFor
\State Sample $t^* \in \{0,\ldots,T\}$ according to $\mathbb{P}(t^* = t) = \frac{\eta_t}{\sum_{i = 0}^T\eta_i}$.
\State \Return $\bm{\theta}_{t^*}$.
\end{algorithmic}
\end{algorithm}

\subsection{Convergence Analysis}
\par We proceed by proving the convergence of Algorithm \ref{Algorithm: ZoSGA}. Let us start by proving certain technical results.
\begin{lemma} \label{lemma: bound on sample gradient variance}
Let Assumption  \textnormal{\ref{assumption: oracle assumption}} hold, and fix any $\bm{\theta} \in \Theta$. For a.e. $\omega \in \Omega$ let $\bm{W}$ be the output of the oracle at $(\bm{\theta},\omega)$. Then, for any $\mu \geq 0$, and any $\bm{U} \sim \mathcal{N}(\bm{0},\bm{I})$, the following holds
\begin{equation} \label{eqn: bound on the expected gradient}
\mathbb{E}\left\{\left\| \bm{D}_{\mu}(\bm{\theta},\omega,\bm{U};\bm{W})\right\|^2\right\} \leq  4B_F^2 L_{H,0}^2(S^2 + 2S),
\end{equation}
\noindent where $\bm{U}$ and $\omega$ are assumed to be statistically independent.
\end{lemma}
\begin{proof}
\par By Lemma \ref{lemma: boundedness of Wirtinger gradient}, we have that for any $(\bm{\theta},\bm{W}) \in \Theta\times \mathcal{W}$ and a.e. $\omega \in \Omega$, $\left\|\frac{\partial^{\circ}}{\partial \bm{z}} F(\bm{W},\bm{z})\big\vert_{\bm{z} = \bm{H}(\bm{\theta},\omega)}\right\| \leq B_F$. Using the Cauchy-Schwartz inequality, we obtain
\begin{equation*}
    \begin{split}
       & \mathbb{E}\left\{\left\|\bm{D}_{\mu}(\bm{\theta},\omega,\bm{U}) \right\|^2 \right\} \leq  \frac{1}{\mu^2} B_F^2 \mathbb{E}\left\{\left\| \left(\bm{\Delta}_{\mu}(\bm{\theta},\omega,\bm{U}) U^\top \right)^\top\right\|^2\right\} \\
      &\qquad  \leq  \frac{1}{\mu^2} B_F^2 \mathbb{E}\left\{ \mathbb{E}\left\{ \|\bm{\Delta}_{\mu}(\bm{\theta},\omega,\bm{U})\|^2 \|\bm{U}\|^2 \big\vert \bm{U}\right\}\right\} \\
      &\qquad  \leq  4 B_F^2 L_{H,0}^2 \mathbb{E}\left\{\|\bm{U}\|^4 \right\} = 4B_F^2 L_{H,0}^2(S^2 + 2S) ,
    \end{split}
\end{equation*}
\noindent where in the last inequality we used Lipschitz continuity of $\bm{H}(\cdot,\omega)$ (from \textbf{(A3)}) while in the last equality we evaluated the 4-th moment of the $\chi$-distribution. Since $F\left(\bm{W},\bm{H}(\bm{\theta},\cdot)\right) \in \mathcal{Z}_2$ from \textbf{(A5)}, we observe that all of the above expectations are well-defined and finite and hence the proof is complete.
\end{proof}

\begin{lemma} \label{lemma: bound on difference between expected gradient and expected estimator}
Let Assumption \textnormal{\ref{assumption: oracle assumption}} hold, fix any $\bm{\theta} \in \Theta$, and for a.e. $\omega \in \Omega$, let $\bm{W}$ be the output of the oracle. Then, the sample gradient approximation given in \textnormal{\eqref{eqn: ZO gradient of the compositional function}} satisfies
\begin{equation*}
\begin{split}
\mathbb{E}\left\{\bm{D}_{\mu}(\bm{\theta},\omega,\bm{U};\bm{W})\right] =& \ \mathbb{E}\left\{\nabla_{\bm{\theta}}^{\mu} F\left(\bm{W},\bm{H}(\bm{\theta},\omega)\right) \right\} \triangleq \widehat{\nabla} f(\bm{\theta}), \\
 \big\| \widehat{\nabla}f(\bm{\theta}) - \nabla f(\bm{\theta}) \big\| \leq &\ 2 \mu B_F L_{H,1} \sqrt{M S}.
\end{split}
\end{equation*}
\end{lemma}
\begin{proof}
Firstly, by utilizing Fubini's theorem we obtain that 
\[\mathbb{E}\left\{\bm{D}_{\mu}(\bm{\theta},\omega,\bm{U};\bm{W})\right] = \mathbb{E}\left\{\nabla_x^{\mu} F\left(\bm{W},\bm{H}(\bm{\theta},\omega)\right) \right\}, \]
\noindent where the first expectation is taken with respect to the product measure of the two random variables $\omega$ and $U$. Furthermore,
\begin{equation*}
    \begin{split}
        & \left\|\nabla_{\bm{\theta}}^{\mu} F\left(\bm{W},\bm{H}(\bm{\theta},\omega)\right) - \nabla_{\bm{\theta}}F\left(\bm{W},\bm{H}(\bm{\theta},\omega)\right)  \right\| \\
         &\quad = \left\| 2 \Re\bigg(\frac{\partial^{\circ}}{\partial \bm{z}} F(\bm{W},\bm{z}) \bigg\vert_{\bm{z} = \bm{H}(\bm{\theta},\omega)} \left(\bm{\varepsilon}^R + j \bm{\varepsilon}^I\right) \bigg)\right\|\\
         &\quad \leq\ 2B_F \left\|\bm{\varepsilon}^R + j \bm{\varepsilon}^I\right\|
    \end{split}
\end{equation*}
\noindent where $\bm{\varepsilon}^R \triangleq \nabla_{\bm{\theta}} \Re\left(\bm{H}(\bm{\theta},\omega)\right) - \nabla_{\bm{\theta}} \bm{H}_{\mu}^R(\bm{\theta},\omega)$ and $\bm{\varepsilon}^I \triangleq \nabla_{\bm{\theta}} \Im\left(\bm{H}(\bm{\theta},\omega)\right) - \nabla_{\bm{\theta}} \bm{H}_{\mu}^I(\bm{\theta},\omega)$, and we used Lemma \ref{lemma: boundedness of Wirtinger gradient}. Next, we observe (from equivalence of norms) that 
\begin{equation*}
    \begin{split}
\|\bm{\varepsilon}^R\| &\leq \sqrt{M} \max_i \left\| \left(\nabla_{\bm{\theta}} \Re\left(\bm{H}(\bm{\theta},\omega)\right) - \nabla_{\bm{\theta}} \bm{H}_{\mu}^R(\bm{\theta},\omega) \right)_i\right\|\\ & \leq   \frac{\mu}{2} L_{H,1}\sqrt{MS}, 
\end{split}
\end{equation*}
\noindent where, given a matrix $\bm{A}$, $(\bm{A})_i$ denotes the $i$-th column vector, and the last bound follows as in the proof of \cite[Theorem 1]{arXiv:Kumar_etal}. The same procedure can be repeated for bounding $\|j\bm{\varepsilon}^I\|$ and hence a simple application of the triangle inequality and subsequently Jensen's inequality yield the desired result.
\end{proof}
\textit{The Moreau envelope:} Let us write the objective function of \eqref{eqn: first-stage problem} as $\phi(\bm{\theta}) \triangleq -f(\bm{\theta}) + \delta_{\Theta}(\bm{\theta})$, where $\delta_{\Theta}(\bm{\theta})$ is the indicator function for the convex set $\Theta$. Given some penalty parameter $\lambda > 0$, we define the proximity operator as 
\[\textbf{prox}_{\lambda \phi}(\bm{u}) \triangleq \underset{\bm{\theta} \in \mathbb{R}^S}{\arg\min} \left\{\phi(\bm{\theta}) + \frac{1}{2\lambda}\|\bm{u}-\bm{\theta}\|^2 \right\}, \]
\noindent and the corresponding Moreau envelope as
\[ \phi^{\lambda}(\bm{u}) \triangleq \min_{\bm{\theta} \in \mathbb{R}^S} \left\{\phi(\bm{\theta}) + \frac{1}{2\lambda}\|\bm{u}-\bm{\theta}\|^2 \right\}.\]
\noindent It is well-known that the Moreau envelope with parameter $\lambda > \rho$ (the weak convexity constant) is smooth even if $\phi(\cdot)$ is not, and the magnitude of its gradient can be used as a near-stationarity measure of the non-smooth problem of interest. Indeed, if a point $\bm{\theta}$ is $\epsilon$-stationary for the Moreau envelope, then it is close to an near-stationary point of \eqref{eqn: two-stage problem}. This (standard) approach is adopted in this work, following \cite{SIAMOpt:Davis,arXiv:Pougk-Kal}, among others.

\begin{theorem} \label{thm: convergence analysis}
Let Assumption \textnormal{\ref{assumption: oracle assumption}} be in effect and assume that $\{\bm{\theta}_t\}_{t = 0}^T$, $T > 0$, is generated by \textnormal{\ref{Algorithm: ZoSGA}}, where $\bm{\theta}_{t^*}$ is the point that the algorithm returns. For any $\bar{\rho} > \rho$, it holds that
\begin{align}
      &  \mathbb{E}\left\{\left\| \nabla \phi^{1/\bar{\rho}}(\bm{\theta}_{t^*})\right\|^2\right\} \label{eqn: returning point Moreau gradient}\\
      &\ \  \leq \frac{\bar{\rho}}{\bar{\rho}-\rho}\Bigg( \frac{\phi^{1/\bar{\rho}}(x_0) - \min \phi(x) +  C_2 \bar{\rho} \sum_{t = 0}^T \eta_t^2 }{\sum_{t=0}^T \eta_t} + C_1 \bar{\rho}\mu\Bigg) \nonumber
\end{align}
\noindent where, letting $\Delta_{\Theta}$ be the diameter of $\Theta$,
\[C_{1} \triangleq 2\Delta_{\Theta} B_F  L_{H,1}\sqrt{M S},\quad C_{2} \triangleq 2B_F^2 L_{H,0}^2(S^2 + 2S).\] 
\noindent Moreover, if we set $\bar{\rho} = 2\rho$, and
\[ \eta_t = \sqrt{\frac{\Delta_f}{2C_2 \rho  (T+1)}},\qquad \textnormal{for all }t\geq 0,\]
\noindent for some $\Delta_f \geq \phi^{1/(2\rho)}(\bm{\theta}_0) - \min  \phi(\bm{\theta})$, then it holds that
\begin{equation} \label{eqn: bound on Moreau envelope}
    \begin{split}
     \mathbb{E}\left\{\left\| \nabla \phi^{1/(2\rho)}(\bm{\theta}_{t^*})\right\|^2\right\} \leq &\ 8\left(\sqrt{\frac{\Delta_f \rho C_2}{2(T+1)}} + C_1\rho \mu\right).
    \end{split}
\end{equation}
\end{theorem}
\begin{proof}
\par For any $t \geq 0$, we have $\widehat{\nabla}f(\bm{\theta}_t) \equiv \mathbb{E}_{[t]}\left\{\bm{D}_{\mu}(\bm{\theta},\omega,\bm{U})\right\}$, where $\mathbb{E}_{[t]}\{\cdot\} \triangleq \mathbb{E}\{\cdot \vert \bm{U}_{t-1},\omega_{t-1},\ldots,\bm{U}_0,\omega_0\}$ (see Lemma \ref{lemma: bound on difference between expected gradient and expected estimator}). We define the point $\hat{\bm{\theta}}_t \triangleq \textbf{prox}_{\phi/\bar{\rho}}(\bm{\theta}_t)$. Then, we obtain
\begin{equation*}
    \begin{split}
       & \mathbb{E}_{[t]}\left\{\phi^{1/\bar{\rho}}(\bm{\theta}_{t+1})\right\} \leq  \mathbb{E}_{[t]}\left\{f(\bar{\bm{\theta}}_{t}) +  \frac{\bar{\rho}}{2}\big\|\hat{\bm{\theta}}_t - \bm{\theta}_{t+1}\big\|^2\right\}\\
      &\quad  = \phi\big(\hat{\bm{\theta}}_t\big) + \frac{\bar{\rho}}{2}\mathbb{E}_{[t]}\bigg\{\big\| \textbf{proj}_{\Theta}(\hat{\bm{\theta}}_t)\\ &\qquad \qquad      -\textbf{proj}_{\Theta}\left(\bm{\theta}_t + \eta_t \bm{D}_{\mu}(\bm{\theta}_t,\omega_t,\bm{U}_t) \right) \big\|^2 \bigg\} \\
       &\quad \leq  \phi\big(\hat{\bm{\theta}}_t\big) + \frac{\bar{\rho}}{2}\mathbb{E}_{[t]}\left\{\left\|\bm{\theta}_t + \eta_t \bm{D}_{\mu}(\bm{\theta}_t,\omega_t,\bm{U}_t)  - \hat{\bm{\theta}}_t \right\|^2 \right\} \\
     &\quad  \leq  \phi\big(\hat{\bm{\theta}}_t\big) + \frac{\bar{\rho}}{2}\big\|\bm{\theta}_t - \hat{\bm{\theta}}_t\big\|^2 \\
       &\qquad  + \bar{\rho}\eta_t\mathbb{E}_{[t]}\left\{\left\langle \hat{\bm{\theta}}_t - \bm{\theta}_t, -\bm{D}_{\mu}(\bm{\theta}_t,\omega_t,\bm{U}_t)\right\rangle \right\} + C_2\bar{\rho}\eta_t^2  \\       
      &\quad =  \phi\big(\hat{\bm{\theta}}_t\big) + \frac{\bar{\rho}}{2}\big\|\bm{\theta}_t - \hat{\bm{\theta}}_t\big\|^2 + \bar{\rho}\eta_t \left\langle \hat{\bm{\theta}}_t - \bm{\theta}_t, -{\nabla} f(\bm{\theta})\right\rangle  \\ &\qquad + \bar{\rho} \eta_t\left\langle \hat{\bm{\theta}}_t - \bm{\theta}_t,{\nabla}f(\bm{\theta}_t)-\widehat{\nabla}f(\bm{\theta}_t) \right\rangle  +  C_2\bar{\rho}\eta_t^2 \\
     &\quad \leq  \phi^{1/\bar{\rho}}(\bm{\theta}_t) + \bar{\rho}\eta_t\left\langle \hat{\bm{\theta}}_t - \bm{\theta}_t, -\nabla f(\bm{\theta}) \right\rangle \\ &\qquad + \bar{\rho} \eta_t \left\|\hat{\bm{\theta}}_t - \bm{\theta}_t \right\| \left\|\widehat{\nabla} f(\bm{\theta}) - \nabla f(\bm{\theta}) \right\| + C_2 \bar{\rho}\eta_t^2\\
       &\quad \leq  \phi^{1/\bar{\rho}}(\bm{\theta}_t) + \bar{\rho}\eta_t\left(f(\bm{\theta}_t)-f\big(\hat{\bm{\theta}}_t\big) + \frac{\rho}{2}\left\|\bm{\theta}_t - \hat{\bm{\theta}}_t\right\|^2\right) \\
       &\qquad +C_1\bar{\rho} \mu \eta_t +  C_2\bar{\rho}\eta_t^2,
    \end{split}
\end{equation*}
\noindent where in the second inequality we used the non-expansiveness of the projection, in the fourth inequality we used Cauchy-Schwartz as well as the definition of the Moreau envelope, and in the fifth inequality we used the weak convexity of $-f(\cdot)$, Lemma \ref{lemma: bound on difference between expected gradient and expected estimator} and the fact that $\Theta$ is assumed to be compact (and hence there exists a constant $\Delta_{\Theta} > 0$ such that $\big\|\hat{\bm{\theta}}_t- \bm{\theta}_t\| \leq \Delta_{\Theta}$). Next, by following exactly the developments in \cite[Section 3.1]{SIAMOpt:Davis}, we notice that the mapping $\bm{\theta} \mapsto -f(\bm{\theta}) + \frac{\bar{\rho}}{2}\|\bm{\theta}-\bm{\theta}_t\|^2$ is strongly convex with parameter $\bar{\rho}-\rho$, and is minimized at $\hat{\bm{\theta}}_t$, thus we obtain
\begin{equation*}
\begin{split}
& f(\hat{\bm{\theta}}_t) - f(\bm{\theta}_t)  - \frac{\rho}{2}\|\bm{\theta}_t - \hat{\bm{\theta}}_t\|^2 =  \left(- f(\bm{\theta}_t)  + \frac{\bar{\rho}}{2}\left\|\hat{\bm{\theta}}_t - \hat{\bm{\theta}}_t\right\|^2\right)   \\ &\qquad -\left(-f(\hat{\bm{\theta}}_t) + \frac{\bar{\rho}}{2}\left\|{\bm{\theta}}_t - \hat{\bm{\theta}}_t\right\|^2 \right)  + \frac{\bar{\rho}-\rho}{2}\left\|\bm{\theta}_t - \hat{\bm{\theta}}_t\right\|^2  \\
& \quad \geq (\bar{\rho}-\rho)\left\|\bm{\theta}_t - \hat{\bm{\theta}}_t\right\|^2 \equiv \frac{\bar{\rho}-\rho}{\bar{\rho}^2}\left\|\nabla \phi^{1/\bar{\rho}}(\bm{\theta}_t) \right\|^2, 
\end{split}
\end{equation*}
\noindent where the last equivalence follows from \cite[Lemma 2.2]{SIAMOpt:Davis}. Hence,
\begin{equation*}
    \begin{split}
       \mathbb{E}_{[t]}\left\{\phi^{1/\bar{\rho}}(\bm{\theta}_{t+1})\right\} \leq &\   \phi^{1/\bar{\rho}}(\bm{\theta}_t) - \frac{\eta_t (\bar{\rho}-\rho)}{\bar{\rho}}\left\| \nabla \phi^{1/\bar{\rho}}(\bm{\theta}_t)\right\|^2 \\&\quad + C_1 \bar{\rho}\mu\eta_t + C_2\bar{\rho}\eta_t^2.
    \end{split}
\end{equation*}
\noindent Taking expectations with respect to the history $\omega_0$, $\bm{U}_0$, $\ldots,\ \omega_{t-1}$, $\bm{U}_{t-1}$ and using the total expectation, yields 
\begin{equation*} \label{eqn: convergence analysis Moreau envelope expected descent}
\begin{split}
    & \mathbb{E}\left\{ \phi^{1/\bar{\rho}}(\bm{\theta}_{t+1}) \right\} \leq   \mathbb{E}\left\{\phi^{1/\bar{\rho}}(\bm{\theta}_t)\right\} + \bar{\rho}\eta_t(\mu C_1 + C_2\eta_t)  \\   
   &\qquad\qquad\qquad\qquad\qquad  -  \frac{\eta_t(\bar{\rho}-\rho)}{\bar{\rho}} \mathbb{E}\left\{\left\| \nabla \phi^{1/\bar{\rho}} (\bm{\theta}_t)\right\|^2\right\},
     \end{split}
\end{equation*}
\noindent Subsequently, we can unfold the latter inequality to obtain
\begin{equation*}
    \begin{split}
        \mathbb{E}\left\{\phi^{1/\bar{\rho}}(\bm{\theta}_{T+1})\right\} \leq\ & \phi^{1/\bar{\rho}}(\bm{\theta}_0) + C_1\bar{\rho} \mu  \sum_{t=0}^T \eta_t +C_2 \bar{\rho} \sum_{t = 0}^T \eta_t^2\\&\quad\quad  - \frac{\bar{\rho}-\rho}{\bar{\rho}} \sum_{t=0}^T \eta_t\mathbb{E}\left\{\left\| \nabla \phi^{1/\bar{\rho}}(\bm{\theta}_t)\right\|^2\right\}.
    \end{split}
\end{equation*}
\noindent Then, we can lower bound the left-hand side by $\phi(\bm{\theta}^*) \triangleq \min_{\bm{\theta} \in \Theta} f(\bm{\theta})$, and rearrange, to obtain
\begin{equation*}
    \begin{split}
   & \frac{1}{\sum_{t=0}^T \eta_t} \sum_{t=0}^T \eta_t \mathbb{E}\left\{\left\| \nabla \phi^{1/\bar{\rho}}(\bm{\theta}_t)\right\|^2\right\} \\
   &\quad \leq  \frac{\bar{\rho}}{\bar{\rho}-\rho}\Bigg( \frac{\phi^{1/\bar{\rho}}(\bm{\theta}_0) - \phi(\bm{\theta}^*) + C_2\bar{\rho} \sum_{t = 0}^T \eta_t^2 }{\sum_{t=0}^T \eta_t} + C_1\bar{\rho} \mu\Bigg).
    \end{split}
\end{equation*}
\noindent Since the left-hand side is exactly $\mathbb{E}\{\| \nabla \phi^{1/\bar{\rho}}(\bm{\theta}_{t^*})\|^2\}$, we deduce that \eqref{eqn: returning point Moreau gradient} holds. Finally, setting $\bar{\rho} = 2\rho$, letting ${\Delta_f} \geq \phi^{1/\bar{\rho}}(\bm{\theta}_0) - \min \phi(\bm{\theta})$, and choosing a constant step size as
\[ \eta_t = \sqrt{\frac{\Delta_f}{2C_2(T+1)}},\qquad \textnormal{for all }t\geq 0,\]
\noindent we obtain \eqref{eqn: bound on Moreau envelope} which completes the proof.
\end{proof}
\begin{remark} \label{remark: complexity}
Note that choosing 
$\mu = \mathcal{O}\big(1/
\sqrt{(M T)}\big)$
yields that 
$\mathbb{E}\big\{\big\Vert \nabla \varphi^{1/(2\rho)}\big(\bm{\theta}_{t^*}\big)\big\Vert
     \big\} \leq \epsilon$, 
after $\mathcal{O}\left(\sqrt{S}\epsilon^{-4}\right)$ iterations.
\end{remark}

\section{Case Study: Sumrate Maximization}\label{sec: case study}
Capitalizing on a standard IRS-aided MISO downlink scenario (see Figure \ref{fig:env_setup} in Section \ref{sec: Simulations}), our goal here is to maximize the total downlink rate of $K$ users actively serviced by an AP with $M$ antennas, while passively aided by one or multiple IRSs, arbitrarily spatially placed. As usual, we assume dynamic (reactive) AP precoders, while the IRS beamformers are static (non-reactive) tunable elements.
We make no sensing assumptions on the IRSs, i.e., the IRSs are completely passive network elements.

Each of the users $k=1,\ldots,K$ experiences a random effective channel denoted by $\bm{h}_k\left(\bm{\theta},\omega \right)$, indexed by the IRS parameter vector $\bm{\theta}$ as well as the usual state of nature $\omega\in \Omega$ describing \textit{unobservable} random propagation patterns for each value of $\bm{\theta}$. In other words, $\bm{h}_k\left(\bm{\theta},\omega \right)$ is a random channel field with spatial variable $\bm{\theta}$. We make the standard assumption that the effective channels $\bm{h}_k\left(\bm{\theta},\omega \right),k=1,\ldots,K$, are known to the AP at the time of transmission \cite{sca:guo2020larsson,sca:zhao2020tts}. Note that the implementation complexity of estimating effective channels in our setting is exactly the same as that in conventional multi-user downlink beamforming (i.e., involving no IRSs), \textit{regardless} of the number and/or spatial configuration of the IRSs; no extra hardware or customized scheduling schemes are required on either the AP or the IRSs assisting the network.
\par The QoS of user $k$ is measured by the corresponding SINR,
\begin{align}\nonumber
\begin{aligned} \label{eq:2}
\text{SINR}_k(\bm{W},\bm{h}_k(\bm{\theta}, \omega )) \triangleq \frac{\left|\bm{h}_k^\hermtr(\bm{\theta}, \omega)\bm{w}_k\right|^2}{{\sum}_{{j \in \mathbb{N}_{K}^+ \setminus k}}\left|\bm{h}^\hermtr_k(\bm{\theta}, \omega)\bm{w}_j\right|^2 + \sigma^2_k},
\end{aligned}
\end{align}
where $\bm{W} {=} \mathrm{vec}(\begingroup \setlength\arraycolsep{2pt} \begin{bmatrix} \bm{w}_1 & \bm{w}_2 & \cdots & \bm{w}_K \end{bmatrix} \endgroup) \in \mathbb{C}^{M_U\triangleq (M \times K)}$, $\bm{w}_k$ is a transmit precoding vector and $\sigma^2_k$ is the noise variance for user $k$, respectively. Then, the \textit{weighted sumrate utility} of the network is defined as
\begin{equation}\nonumber
{F}(\bm{W},\bm{H}(\bm{\theta},\omega)) \triangleq \sum_{k=1}^K \alpha_k \log_2\left(1 + \textnormal{SINR}_k\left(\bm{W},\bm{h}_k\left(\bm{\theta},\omega  \right)\right)\right),
\end{equation}
\noindent with $\bm{H} = \mathrm{vec}(\begin{bmatrix} \bm{h}_1 \ldots \bm{h}_K \end{bmatrix})\in \mathbb{C}^{M_U}$, and $\alpha_k\ge0$ the weight associated with user $k$.
We are interested in maximizing the sumrate of the network jointly by selecting instantaneous-optimal dynamic AP precoders $\bm{W}$, and on-average-optimal static IRS beamformers $\bm{\theta}$ \cite{sca:guo2020larsson,sca:zhao2020tts}, i.e., we are interested in the problem
\begin{equation}\label{eqn: sumrate two-stage problem} \tag{2SSRM}
\max_{\bm{\theta} \in \Theta} \mathbb{E} 
\bigg\{ 
\max_{\bm{W}: \|\bm{W}\|^2 \leq P} 
{F}\left(\bm{W},\bm{H}(\bm{\theta},\omega)\right)
\bigg\},
\end{equation}
where $P>0$ is a total power budget at the AP, and $\Theta$ is a real convex and compact feasible set of amplitudes and phases. Problem \eqref{eqn: sumrate two-stage problem} is an instance of \eqref{eqn: two-stage problem}. 
\par \textit{Assumption Compatibility:} Let us now briefly discuss the compatibility of Assumption \ref{assumption: two-stage problem} with problem \eqref{eqn: sumrate two-stage problem}. We firstly note that conditions \textbf{(A1)}, \textbf{(A2)} are both satisfied. Condition \textbf{(A5)} is also satisfied in light of uniform boundedness. Condition \textbf{(A3)} of Assumption \ref{assumption: two-stage problem} imposes regularity that is required for the grounded development of our optimization scheme and for its convergence analysis. Also, observe that the boundedness assumption in condition \textbf{(A3)} is natural, since (IRS-aided) wireless channels are always bounded in practice. While we usually have no information on the analytical properties of the effective channel, this condition is easily satisfied in widely used channel models of IRS-aided systems, see, e.g., \cite{sca:zhao2020tts} or Section \ref{sec: Simulations}.
\par We next showcase that the regularity condition \textbf{(A4)} also holds under several reasonable circumstances. In particular, we identify three typical situations under which condition \textbf{(A4)} (i.e. weak concavity of the sample objective function of the first-stage problem \eqref{eqn: first-stage problem}) is readily satisfied. These will be stated here for completeness, and the reader is referred to Appendix \ref{apdx: Weak concavity} for a technical discussion showcasing how weak concavity can be shown in each of the following cases. 
\par Firstly, condition \textbf{(A4)} is readily satisfied in cases where the second-stage problem \eqref{eqn: second-stage problem} admits a unique solution, assuming, of course, the \emph{strong second-order sufficient optimality conditions} for the second-stage problem (we refer the reader to Appendix \ref{apdx: Weak concavity} for a precise description of these conditions). Indeed, in this case one can invoke the Implicit Function Theorem \cite[Theorem 1B.1]{Springer:DonRock} to showcase twice-continuous differentiability of $f(\cdot)$ (as well as of $\max_{\bm{W} \in \mathcal{W}} F(\bm{W},\bm{H}(\bm{\theta},\omega))$, for a.e. $\omega \in \Omega$), and since $f$ is considered on a compact set, weak concavity then follows immediately. On the other hand, if the solution set of the second-stage problem is connected (instead of a singleton), and instead of the strong second-order sufficient conditions, the problem satisfies the regularity condition given in \cite[Assumption 4]{MathOR:Shapiro}, we also obtain twice-continuous differentiability and thus weak concavity. Finally, if $\bm{H}(\cdot,\omega)$ is real-analytic for a.e. $\omega \in \Omega$ (which often holds for channel models appearing in the literature; e.g. see \cite{sca:zhao2020tts} and Section \ref{sec: Simulations}), then one can show that the function $\max_{\bm{W} \in \mathcal{W}} F(\bm{W},\bm{H}(\bm{\theta},\omega))$ is \emph{sub-analytic}, and satisfies the \emph{\L{}ojasiewicz inequality} (see \cite{LojasiewiczInequality}) at every $\bm{\theta} \in \Theta$, with uniform exponent. If the \L{}ojasiewicz constant is uniformly bounded and the second-stage problem satisfies the strong second-order sufficient optimality conditions, weak concavity also follows (see Appendix \ref{apdx: Weak concavity} for the technical details).

\subsection{Practical Per-Iteration Complexity for Sumrate Maximization} \label{subsec: complexity_analysis}
Before proceeding with the associated simulations of our case study, it would be useful to calculate the practical \emph{per-iteration complexity} of ZoSGA. To that end, we assume that the (deterministic here-and-now) second-stage problem is (approximately) solved using $T_2$ iterations of the well-known WMMSE algorithm \cite{wmmseShi2011}, which has a documented computational complexity given by $\mathcal{O}(T_2 K^2 M^3)$. For the first-stage optimization step, the complexity of the zeroth-order stochastic gradient approximations can be shown to be of the order of $\mathcal{O}(K^2 M + K M^2 + S)$, where $S$ is the number of IRS phase-shift parameters. Thus, the effective per-iteration complexity of ZoSGA is of order $\mathcal{O}(T_2 K^2 M^3 + S)$.
\par As mentioned in Section \ref{sec: intro}, model-based SSCO methods have been recently considered for tackling problem \eqref{eqn: sumrate two-stage problem}, mainly due to their fast convergence. Then, it would be appropriate to compare the per-iteration complexity of ZoSGA with that of a state-of-the-art SSCO method, in particular TTS-SSCO \cite{sca:zhao2020tts}.
\par In order to derive the practical \textit{per-iteration} complexity of TTS-SSCO, we need to assume that the effective channel $\bm{H}(\cdot,\omega)$ follows a specific and common cascaded model
\cite{sca:zhao2020tts}. Considering this, the complexity of TTS-SSCO can be shown to be of order $\mathcal{O}(T_H (T_2 K^2 M^3 + K^2S^2M))$, since each TTS-SSCO iteration requires $T_H$ WMMSE runs for gradient statistical approximation. We should also emphasize that this complexity is \textit{in addition to} the (highly nontrivial) computational effort required for estimating \textit{cascaded} CSI statistics (called S-CSI in \cite{sca:zhao2020tts}), which are required by TTS-SSCO.
%
It is evident from that the proposed algorithm (i.e. ZoSGA) has a significantly smaller per-iteration complexity.

\section{Simulations} \label{sec: Simulations}
Building upon the IRS-aided MISO downlink sumrate maximization case discussed in Section \ref{sec: case study}, we present a set of detailed simulations to empirically evaluate and confirm the efficacy of the proposed ZoSGA algorithm. Unless stated otherwise, all our empirical results are averaged over $2000$ independent simulations, with an additional $4$-th order Savitzky–Golay filter \cite{fig:savitzky1964smoothing} of length $500$, while $95\%$ confidence intervals are also provided. In what follows, we describe and evaluate three distinct wireless network setups and examine the obtained empirical results. The last set of experiments will involve a physical model for the IRSs, which is briefly described in the beginning of Subsection \ref{subsec:Varactor IRSs}.

\vspace{-2pt}
\subsection{Baseline Channel Model} \label{subsec:ch_model}
With the practically feasible assumption of insufficient angular spread of signals in a scattering environment, we consider both LOS and non-LOS channels. Following \cite{sca:zhao2020tts}, we consider three types of intermediate channels in our simulations, namely, the reflected channel ${\bm{h}}_{r,k}$ 
from an IRS to user $k$, the channel ${\bm{G}}$ from an AP to an IRS, as well as the direct path channel ${\bm{h}}_{d,k}$ from an AP to user $k$. We model these as general spatially correlated Rician fading channels \cite{ch_model:mckay2005general}.
Additionally, we assume that the second-order statistics of the LOS links are identical for all users, due to the large distances between an AP and its served users. 

Concretely, for each user $k$, we define the path loss adjusted versions of these channels, respectively, as
\begin{equation} 
\begin{split} \label{eq:rician_model}
    \breve{\bm{h}}_{r,k} &{\triangleq} \sqrt{\frac{\beta_{Iu}}{1{+}\beta_{Iu}}} \check{\bm{v}}_{r,k} {+} \sqrt{\frac{1}{1{+}\beta_{Iu}}} \bm{\Phi}_{r,k}^{1/2}\bm{v}_{r,k},\\
    \breve{\bm{G}} &{\triangleq} \sqrt{\frac{\beta_{AI}}{1{+}\beta_{AI}}} \check{\bm{F}} {+} \sqrt{\frac{1}{1{+}\beta_{AI}}} \bm{\Phi}_r^{1/2}\bm{F}\bm{\Phi}_d^{1/2},\\
    \breve{\bm{h}}_{d,k} &{\triangleq} \sqrt{\frac{\beta_{Au}}{1{+}\beta_{Au}}} \check{\bm{v}}_{d,k} {+} \sqrt{\frac{1}{1{+}\beta_{Au}}} \bm{\Phi}_d^{1/2}\bm{v}_{d,k},\\
\end{split}
\end{equation}
where $\bm{v}_{r,k} \in \mathbb{C}^{N \times 1}, \bm{F} \in \mathbb{C}^{N \times M}, \text{ and } \bm{v}_{d,k} \in \mathbb{C}^{M \times 1}$ are the instantaneous components (I-CSI), and $\check{\bm{v}}_{r,k} \in \mathbb{C}^{N \times 1}, \check{\bm{F}} \in \mathbb{C}^{N \times M}$, and $\check{\bm{v}}_{d,k} \in \mathbb{C}^{M \times 1}$ are the statistical components (S-CSI) of the above channels, all having i.i.d. circularly symmetric complex Gaussian (CSCG) entries with zero mean and unit variance, and with S-CSI being sampled once per simulation. The dimensions $M$ and $N$ denote the number of AP antennas and the number of passive reflecting elements in an IRS, respectively. The matrices $\bm{\Phi}_{r,k} \in \mathbb{C}^{N \times N}, \bm{\Phi}_d \in \mathbb{C}^{M \times M}$, and $\bm{\Phi}_r \in \mathbb{C}^{N \times N}$ are, in order, the reflected channel correlation matrix, the AP transmit correlation matrix, and the IRS receive correlation matrix. We assume an exponential correlation model for $\bm{\Phi}_d$ \cite{IEEECommLet:Loyka}, expressed as
\begin{align} \label{eq:Phi_d}
    \bm{\Phi}_d(i,j) = \begin{cases}r_d^{j-i},&\text{if}\quad i {\leq} j, \\
\bm{\Phi}_d(j,i),&\text{if}\quad i {>} j,
\end{cases}
\end{align}
where $r_d\in(0,1)$ is the correlation coefficient. The matrices $\bm{\Phi}_r$ and $\bm{\Phi}_{r,k}$ are modeled as Kronecker products \cite{IEEECommLet:ChoiLove} as
\begin{equation*}\label{eq:Phi}
\begin{split}
\bm{\Phi}_r &= \bm{\Phi}^h_r \otimes \bm{\Phi}^v_r, \qquad 
\bm{\Phi}_{r,k} = \bm{\Phi}^h_{r,k} \otimes \bm{\Phi}^v_{r,k},
\end{split}
\end{equation*}
where $\bm{\Phi}^h_r,\ \bm{\Phi}^h_{r,k}$, and $\bm{\Phi}^v_r, \bm{\Phi}^v_{r,k}$, for $k \in \mathbb{N}_K^+$, represent the spatial correlation matrices of the horizontal and vertical domains, respectively, and are all modeled similar to $\bm{\Phi}_d$ in \eqref{eq:Phi_d}, with spatial correlation coefficients $r_r \in (0,1)$ and $r_{r,k} \in (0,1)$ for $\bm{\Phi}_r$, and $\bm{\Phi}_{r,k}$, $k \in \mathbb{N}_K^+$, respectively.

\par The deterministic components, i.e, $\check{\bm{v}}_{r,k}, \check{\bm{F}}$, and $\check{\bm{v}}_{d,k}$ determine the moments of their respective CSIs. Lastly, the real scalar $\beta_{Au}$ denotes the Rician fading factor for the LOS channel, while $\beta_{Iu}$ and $\beta_{AI}$ are the same for the reflected channels. These factors define the relative dominance of I-CSIs and S-CSIs in their respective combined CSIs. Moreover, all intermediate channels in \eqref{eq:rician_model} suffer from an exponential path loss proportional to the path distance. We model this loss as  $L_{\alpha}(d) = \sqrt{C_0 d^{-\alpha}}$, where $d$ is path distance in meters, $\alpha$ is the path loss exponent depending upon the channel being considered, and $C_0$ is the common path-loss for when the path distance is one meter.

\par Considering both path loss and fading, we 
may adopt the standard simplified baseline model for the effective channel $\bm{h}_k(\bm{\theta}, \omega)$ of user $k$ 
in the presence of an AP and \textit{one IRS}, i.e.,
\begin{align*}
    \bm{h}_k(\bm{\theta}, \omega)
    \triangleq
    \underbrace{\bm{G}^\hermtr \text{Diag}(\bm{A} \circ e^{j{\bm{\phi}}}) \bm{h}_{r,k}}_{\text{non-LOS link}} 
    + 
    \underbrace{\bm{h}_{d,k} }_{\text{LOS link}} ,
\end{align*}
where $\bm{h}_{r,k} = L_{\alpha_{Iu}}(d_{Iu,k})\breve{\bm{h}}_{r,k}$, $\bm{G} = L_{\alpha_{AI}}(d_{AI})\breve{\bm{G}}$, and $\bm{h}_{d,k} = L_{\alpha_{Au}}(d_{Au,k})\breve{\bm{h}}_{d,k}$. We may take $\omega=\{\bm{G}, \bm{h}_{r,k}, \bm{h}_{d,k},$ $k\in \mathbb{N}_K^+\}$,
while the IRS parameters $\bm{\theta}$ are represented by amplitude and phase vectors $\bm{A} \in [0,1]^{N}$ and $\bm{\phi} \in [-2\pi,2\pi]^{N}$, respectively \cite{ch_model:wu2019intelligent}.
Adding more IRSs to the system increases the non-LOS terms comprising the effective channel of each user accordingly.
\par For slow moving users, as is generally the case, we assume that the values of S-CSI and the spatial correlation matrices remain fixed throughout the duration of the AP service. Further, we assume that IRS-to-IRS links do not contribute to the signal or the interference in the presented multi-IRS cases. Of course, the latter assumption is only made for ease of presentation. 
\par In Subsections \ref{subsec:Ideal IRS} and \ref{subsec:Varactor IRSs}, having defined the channel model, we present simulated results with ideal and physical IRSs, where physical IRSs are constrained both in terms of amplitude and phase, as well as non-linear sensitivity relative to the action space, i.e., the ranges of varactor diode capacitances \cite{irs_model:costa2021electromagnetic}.

\subsection{Networks with Ideal IRSs} \label{subsec:Ideal IRS}
\begin{figure}
  \centering  \centerline{\hspace{-16pt}\includegraphics[width=3.6in]{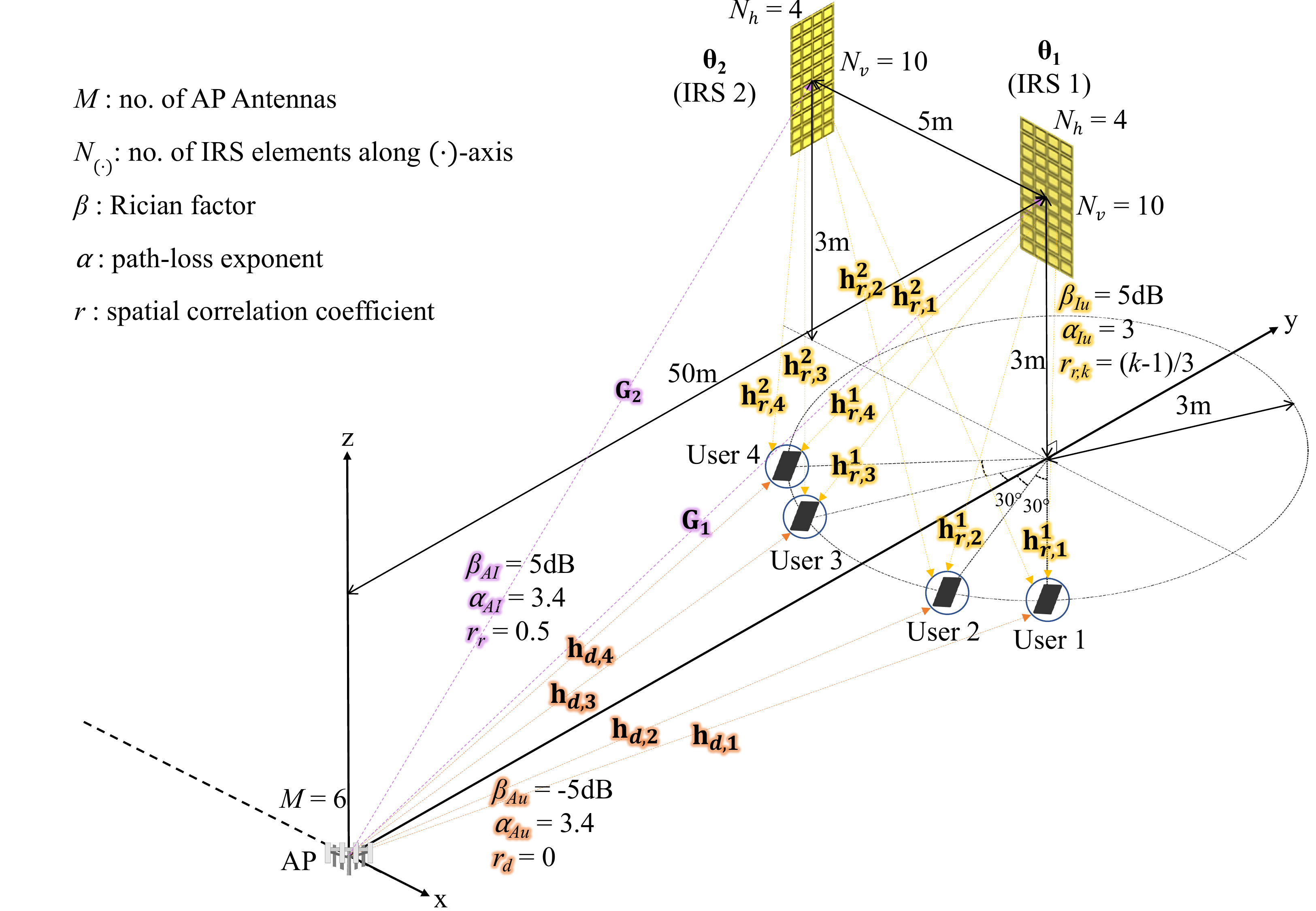}}
\vspace{6bp}
\caption{First IRS-aided network configuration (ideal IRSs).}
\label{fig:env_setup}
\end{figure}

\begin{figure*}[h]
  \centering
  \begin{subfigure}[b]{0.328\textwidth}
  \centering
  \includegraphics[width=\textwidth, height=0.75692307692\textwidth]{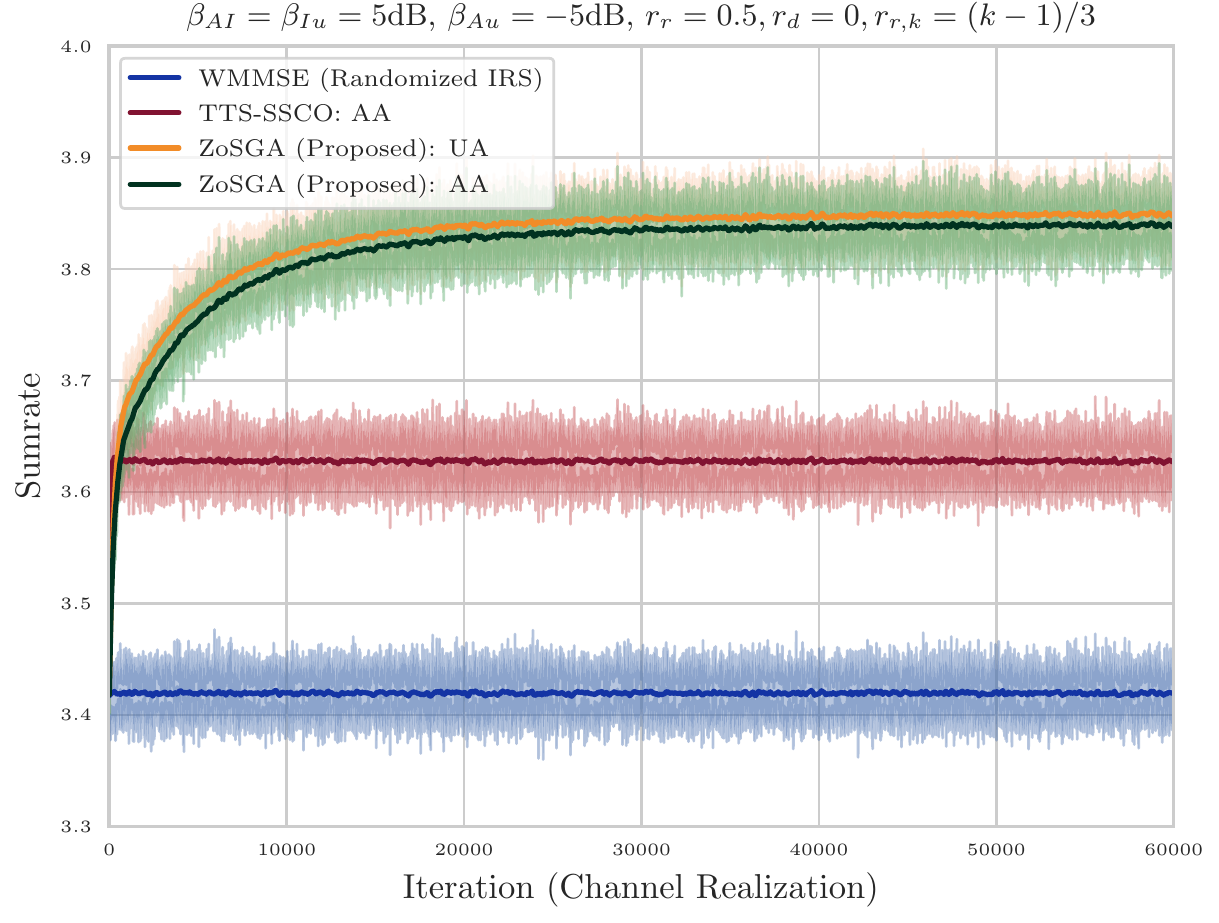}
  \caption{}
  \label{fig:converge1}
  \end{subfigure}
  \begin{subfigure}[b]{0.328\textwidth}
  \centering
  \includegraphics[width=\textwidth, height=0.75692307692\textwidth]{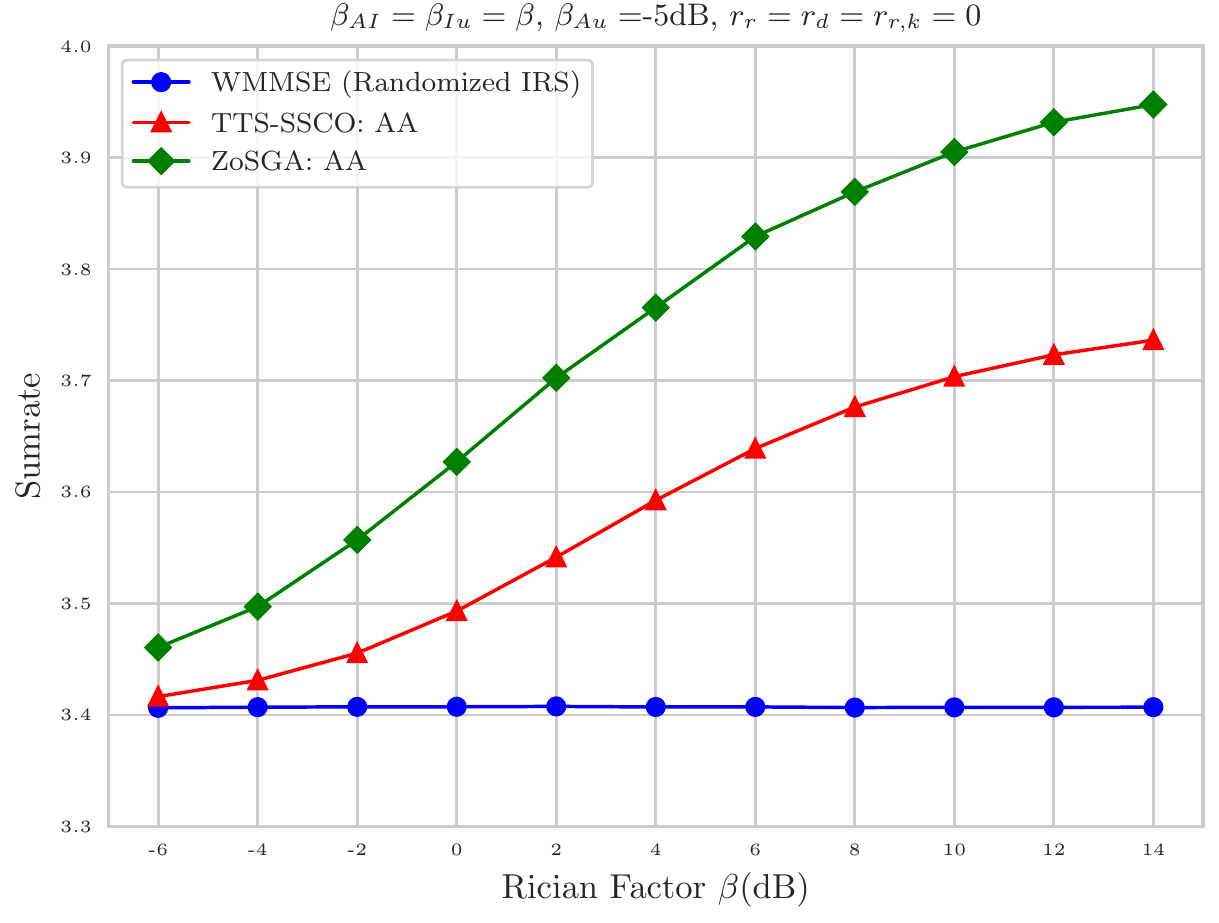}
  \caption{}
  \label{fig:rician}
  \end{subfigure}
  \begin{subfigure}[b]{0.328\textwidth}
  \centering
  \includegraphics[width=\textwidth, height=0.75692307692\textwidth]{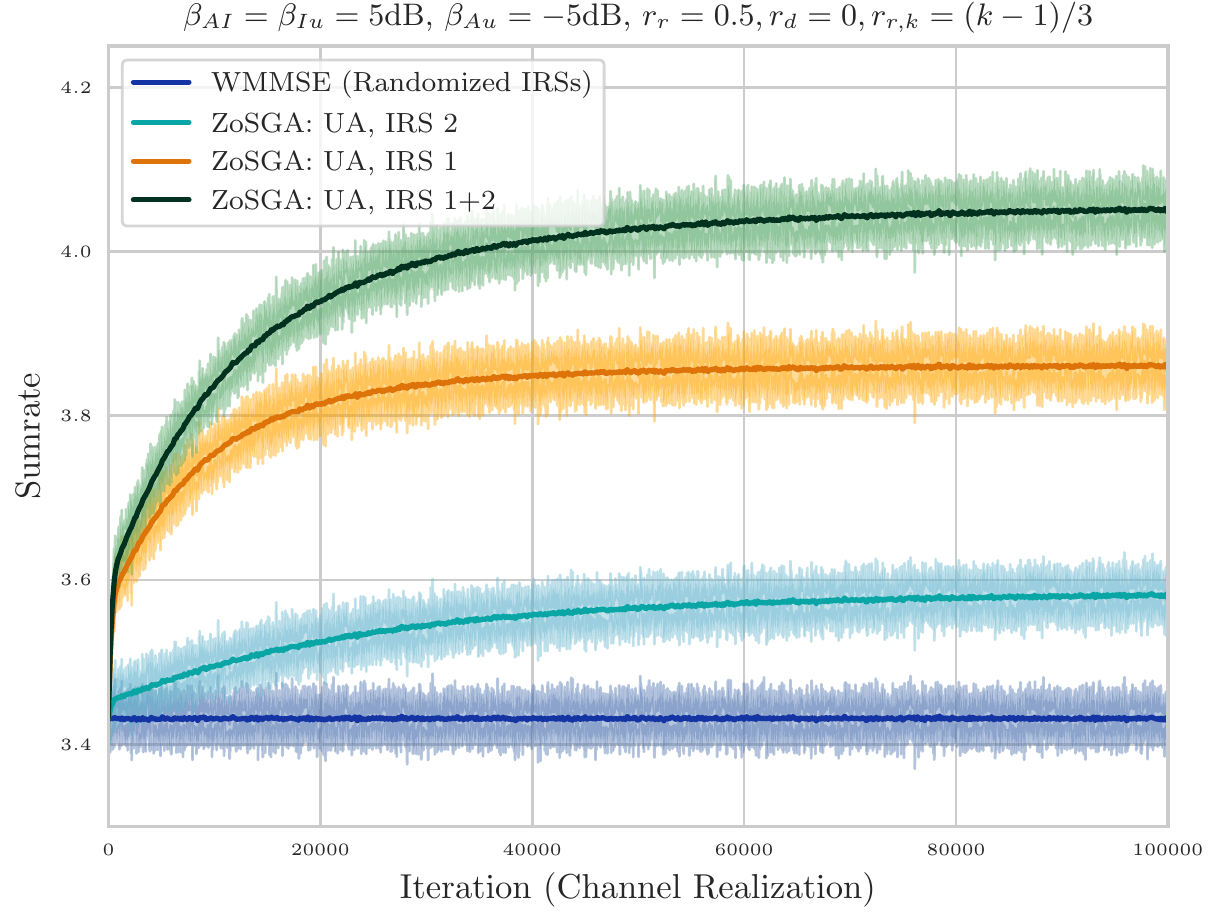}
  \caption{}
  \label{fig:converge2}
  \end{subfigure} 
\caption{(a) Average sumrates achieved by WMMSE \cite{wmmseShi2011} (random IRS phase-shifts), TTS-SSCO \cite{sca:zhao2020tts}, and ZoSGA, with only IRS 1 present (AA: Adjustable Amplitude $|$ UA: Unit Amplitude); (b) Corresponding average sumrates versus the Rician factor; (c) Average sumrates achieved by WMMSE (random IRS) and ZoSGA, with both IRSs present.}
\label{fig:3}
\vspace{-10bp}
\end{figure*}

In the first set of our simulations, we assume that we have full control over the amplitudes and the phases of IRS phase-shift elements, i.e. there is no constraint on achievable phase-amplitude pairs, and we can control these directly. We start by discussing the wireless network setting as shown in Figure \ref{fig:env_setup}. 

\par To highlight the efficacy of the presence of IRSs in a wireless network, we assume a more acute signal attenuation in the LOS path from the AP to the users. Thus, we set $\alpha_{Au}=3.4$, $\alpha_{Iu}=3$ and $\alpha_{AI}=2.2$, where the largest $\alpha_{Au}$ is the LOS path loss exponent while the remaining two are path loss exponents of IRS-User and AP-IRS links, respectively. Moreover, since the distances between the IRS and its served users are relatively small, IRS elements reflect signals with a finite angular spread and a user-location dependent mean angle in practice \cite{ch_est:yin2013coordinated}.
\par We consider two IRSs, as shown in Figure \ref{fig:env_setup}, equipped with $N {=} N_{h} {\times} N_{v}$ rectangular phase-shift elements where $N_{h} (= 4)$ and $N_{v} (= 10)$ denote the number of columns and rows, respectively ($N=40$). For the ideal case we define the controllable parameters of each IRS as a vector $\bm{\theta}_i {=}  \begin{bmatrix} \bm{\phi}_i^\trans & \hspace{-6pt}\bm{A}_i^\trans \end{bmatrix}^\trans$, for $i \in \{1,2\}$, where $\bm{\phi}_i \in [-2\pi,2\pi]^{N}$ and $\bm{A}_i \in [0,1]^{N}$ are phases and amplitudes of the IRS elements, respectively. We do not consider the relative orientations of IRSs, AP and users, as those can be incorporated via rotation offsets, if needed. The effective channel for a user $k$ can, thus, be expressed as 
\begin{align*} \label{eq: sim_model}
    \bm{h}_k(\bm{\theta},{\omega}) =
    \sum_{i=1}^{2} \underbrace{\bm{G}^\hermtr_{i} \text{Diag}(\bm{A}_i \circ e^{j{\bm{\phi}}_i}) \bm{h}^{i}_{r,k}}_{\bm{\theta}_i\text{-non-LoS link}} +
    \underbrace{\bm{h}_{d,k}}_{\text{LoS link}},
\end{align*}
where $\bm{h}^i_{r,k} = L_{\alpha^i_{Iu}}(d_{Iu,k,i})\breve{\bm{h}}^i_{r,k}$, $\bm{G}^i = L_{\alpha^i_{AI}}(d_{AI,i})\breve{\bm{G}}^i$, and $\bm{h}^i_{d,k} = L_{\alpha^i_{Au}}(d_{Au,k,i})\breve{\bm{h}}^i_{d,k}$ for $i \in \{1,2\}$. Also, for brevity, we take $\bm{\theta} {=} (\bm{\theta}_1, \bm{\theta}_2)$ and $\omega=\{\bm{G}_1,\bm{G}_2, \bm{h}^{1}_{r,k}, \bm{h}^{2}_{r,k}, \bm{h}_{d,k},$ $k\in \mathbb{N}_K^+\}$.

After defining the wireless network setting, we move forward with our simulations. In the first simulated comparison we compare the proposed algorithm (ZoSGA) with the well-documented Stochastic Successive Convex Optimization (SSCO) method, specifically a version of it proposed in \cite{sca:zhao2020tts}, which we shall refer to as Two-Time Scale SSCO (TTS-SSCO). Both ZoSGA and TTS-SSCO employ the WMMSE algorithm to optimize the precoding vectors. To keep the comparison justified, we let WMMSE optimize for 20 iterations per channel instance for both of these methods, and have also included WMMSE with randomized IRS parameters as a reference. All parameters pertaining to the TTS-SSCO are taken from \cite[Section V]{sca:zhao2020tts}, i.e., $T_H = 10$, $\tau = 0.01$, $\rho_t = t^{-0.8}$ and $\gamma_t = t^{-1}$. There are four users, so $k \in \mathbb{N}^+_{4}$, weighted uniformly i.e., $\alpha_k = 1$; the AP has six antennas ($M=6$), and the noise variance is $\sigma_k = -80$dBm for all $k$. The reference path loss is $C_0 = -30$dB, and the total allocated power budget is $P = 5$dBm. The LOS Rician factor is $\beta_{Au} = -5$dB, while $\beta_{Iu} = \beta_{AI} = 5$dB, unless specified otherwise. We let the smoothing parameter $\mu= 10^{-12}$ in Algorithm \ref{Algorithm: ZoSGA}, and choose separate initial step-sizes $\eta^0_{\phi} = 0.4$ and $\eta^0_{A} = 0.01$ for updating phases and amplitudes, respectively, scaled by ${0.9972}^t$ for $t \in \mathbb{N}^+_{10^3}$, keeping them constant for $t > 10^3$.

\par The comparison of ZoSGA and TTS-SSCO in terms of achieved sumrate is shown in Figure \ref{fig:converge1}. The comparison is done with only IRS 1 present, matching the TTS-SSCO setting in \cite{sca:zhao2020tts}, which requires exact channel and network models to work. We observe that ZoSGA, although requiring more iterations to converge, substantially outperforms TTS-SSCO solely on the basis of effective CSI, while having no access to the statistical model of the channel or the spatial configuration of the system. There are two main reasons for this. Firstly, TTS-SSCO evaluates gradients by utilizing internally sampled I-CSI and the corresponding optimal precoding vectors (via WMMSE), with a frequency of ten samples per iteration ($T_H$). This can greatly limit the convergence of TTS-SSCO if the internal channel model is not accurate. Secondly, the surrogate objective employed by TTS-SSCO, which simplifies the problem by decoupling the IRS phase-shift elements, allowing the computation of their optimal values in a closed form, is not equivalent to the original nonconvex sumrate optimization problem. As shown by the convergence curves provided in Figure \ref{fig:converge1}, this introduces additional errors, preventing TTS-SSCO from realizing a competitive QoS gain, compared with ZoSGA.


In Figure \:\ref{fig:rician}, we discern the effect of the Rician factor for spatially uncorrelated channels, i.e., we set $r_d = r_r = r_{r,k} = 0$, $\forall\ k \in \mathbb{N}^+_{4}$. We observe that the relative gain of ZoSGA in the achievable sumrate increases with respect to the Rician factors pertaining to $\bm{\theta}_1$-reflected links, i.e., as we move from I-CSI to S-CSI dominated channels. The gain difference between ZoSGA and TTS-SSCO increases in $\beta$, since in a close-to-deterministic effective channel zeroth-order gradient approximations approach the true gradients, while TTS-SSCO is optimizing an approximate surrogate objective. We also observe that WMMSE with a randomized IRS remains insensitive to changes in the Rician factors; this is expected as the reflected channel is not optimized to gain any performance improvements.

\begin{figure}
\centering
\centerline{\includegraphics[width=3.5in]{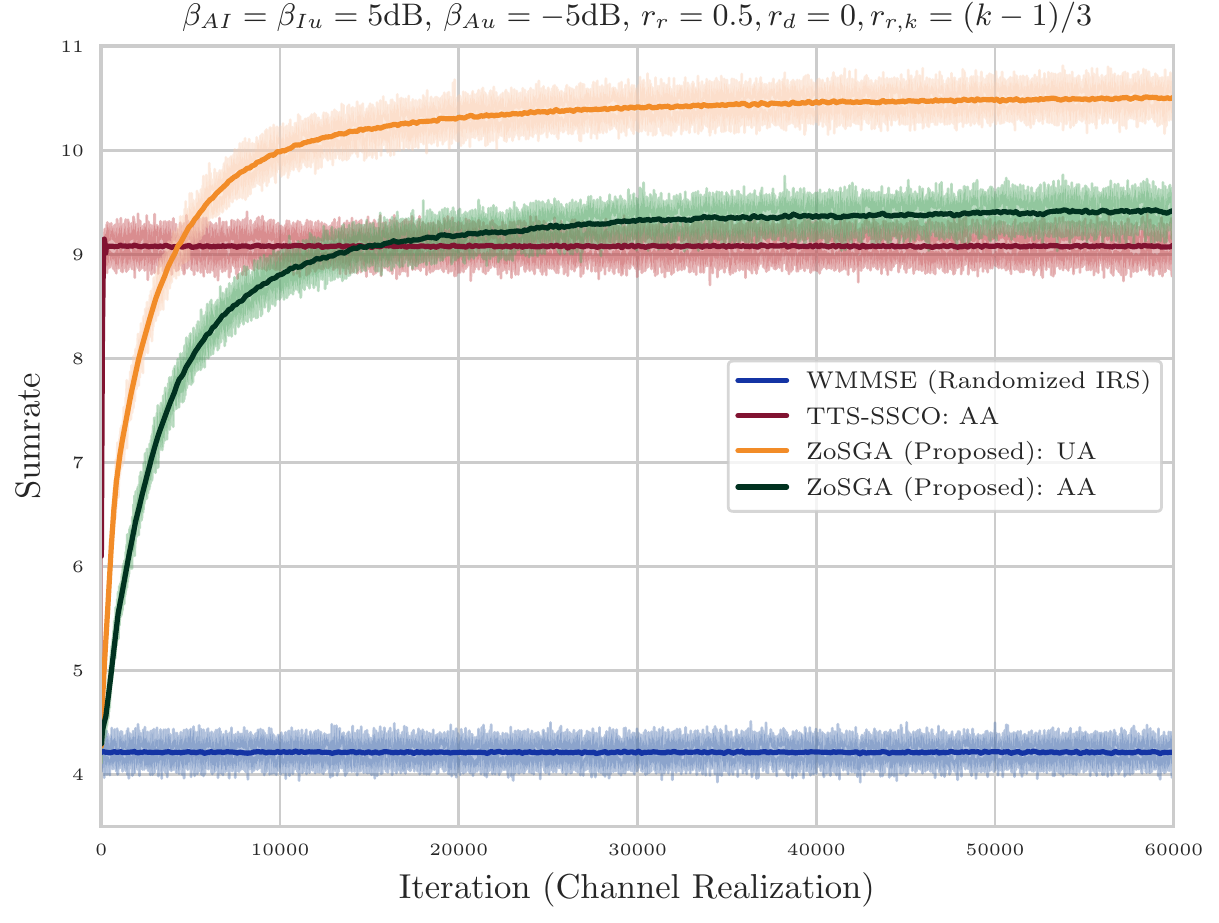}}
\vspace{6bp}
\caption{Average sumrates achieved by WMMSE (random IRS), TTS-SSCO,  and ZoSGA, with only IRS 1 and $N_v=400$.}
\label{fig:converge1_400}
\end{figure}
We can now verify that ZoSGA can also optimize networks with multiple IRSs without any model knowledge. We do so by tuning both IRSs as shown in Figure \ref{fig:env_setup}, while keeping the same hyper-parameters for ZoSGA as well as for WMMSE. Figure \ref{fig:converge2} shows that ZoSGA succeeds in optimizing both IRSs simultaneously, without any information about their spatial configuration. Figure \ref{fig:converge2} demonstrates the improved performance gains when optimizing $\bm{\theta}_1$ versus optimizing the more distant $\bm{\theta}_2$, showing that ZoSGA not only scales well to unknown system/channel models, but is also robust with respect to  $\eta_{\phi}$ and $\eta_{A}$.

\par In practical scenarios, IRSs may have multiple phase-shift elements, frequently on the order of hundreds. To demonstrate that ZoSGA can scale well with the number of IRS parameters $\theta$, we increase $N_v$ to 100 rows of elements, and compare the sumrates achieved by ZoSGA and TTS-SSCO in Figure \ref{fig:converge1_400}. It is evident that ZoSGA is able to scale well and also retains the gain in performance over TTS-SSCO despite the substantial increase in the number of IRS phase-shift elements.

\par Lastly, while in general ZoSGA does take more iterations to converge as compared with model-based methods (here, TTS-SSCO; this is in line with the related literature on model-free stochastic resource allocation: see, e.g., \cite{2019eisen_pfo}),
we may also readily observe that it exhibits a high performance ceiling. Therefore, it can provide very competitive IRS beamformers way before its actual convergence.

\subsection{Networks with Physical IRSs} \label{subsec:Varactor IRSs}
In purely simulated environments, ZoSGA outperforms TTS-SSCO, allowing us to claim it to be a new SOTA for IRS-aided sumrate optimization. Nonetheless, we would also like to evaluate its robustness in a practical setting with \textit{physically modeled} IRSs. To that end, we first define the IRSs using a practically feasible Transmission Line (TL) equivalent of an electromagnetic (EM) model, as presented in \cite{irs_model:costa2021electromagnetic}. This TL model accounts for the geometrical and electrical properties of the IRS elements, also referred to as patches. Specifically, it considers reconfigurability, changes in response due to different angles of wave incidence, mutual coupling among closely spaced cells, and reflection losses.
\par An IRS can be classified as a spatially dispersive device, the reconfigurability of which is achieved by incorporating varactor diodes in its periodic structure.
That is, the beamforming angles (of reflection) are controlled by changing DC voltages, thus tuning the capacitances of these diodes. Though there can be two scenarios when an EM wave impinges on an IRS surface, namely, the traverse magnetic (TM) incidence and the traverse electric (TE) incidence, here we only consider the former for simplicity, mostly focusing on the empirical performance analysis. Moreover, on a side note, we employ the Floquet theorem (as assumed in \cite{irs_model:costa2021electromagnetic}) so as to consider the periodic patches placed in an infinite array, with each element/patch behaving identically. 

\par Given a non-reconfigurable impedance surface, the TL model usually consists of a parallelly connected surface impedance of a reflecting surface, $Z_{surf}$, and inductive impedance of the grounded dielectric slab ($Z_d^{TM}$). Then, the input impedance is $Z_v = Z_{surf} \parallel Z_d$. For an IRS, however, the $Z_{surf}$ is further comprised of a parallel connection between the unloaded surface impedance of the patch array $Z_{patch}$ and the lumped impedance of a varactor diode $Z_{var}$, which is represented as a series of a resistor, inductance and capacitance as 
\begin{equation*}
    Z_{var} = R_{var} + j\omega L_{var} + j \frac{1}{\omega C_{var}},
\end{equation*}
\noindent where the inductance $L_{var}$ depends on the size of the lumped component and must be included in the varactor model to take into account the self-resonance of the component. The resistance $R_{var}$ is included to account for the losses of the varactor. The variable capacitance $C_{var}$ of the diode is used to vary $Z_{var}$. 

\begin{figure}[h]
  \centering
  \centerline{\includegraphics[width=3.6in]{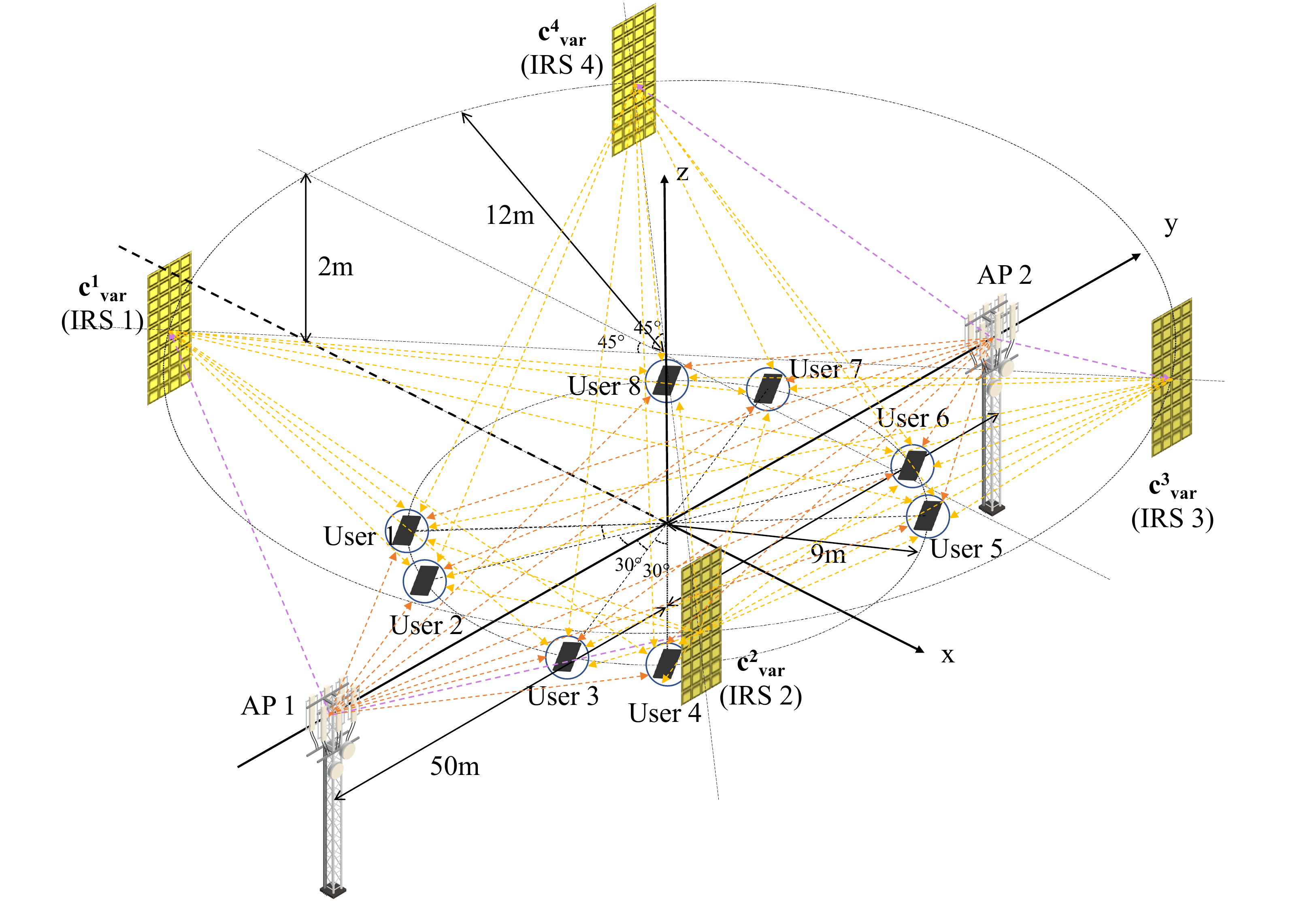}}
\vspace{6bp}
\caption{Second IRS-aided network configuration (physical IRSs).}
\label{fig:env_setup_2}
\end{figure}

\begin{figure*}[h]
  \centering
  \begin{subfigure}[b]{0.328\textwidth}
  \centering
  \includegraphics[width=\textwidth, height=0.75692307692\textwidth]{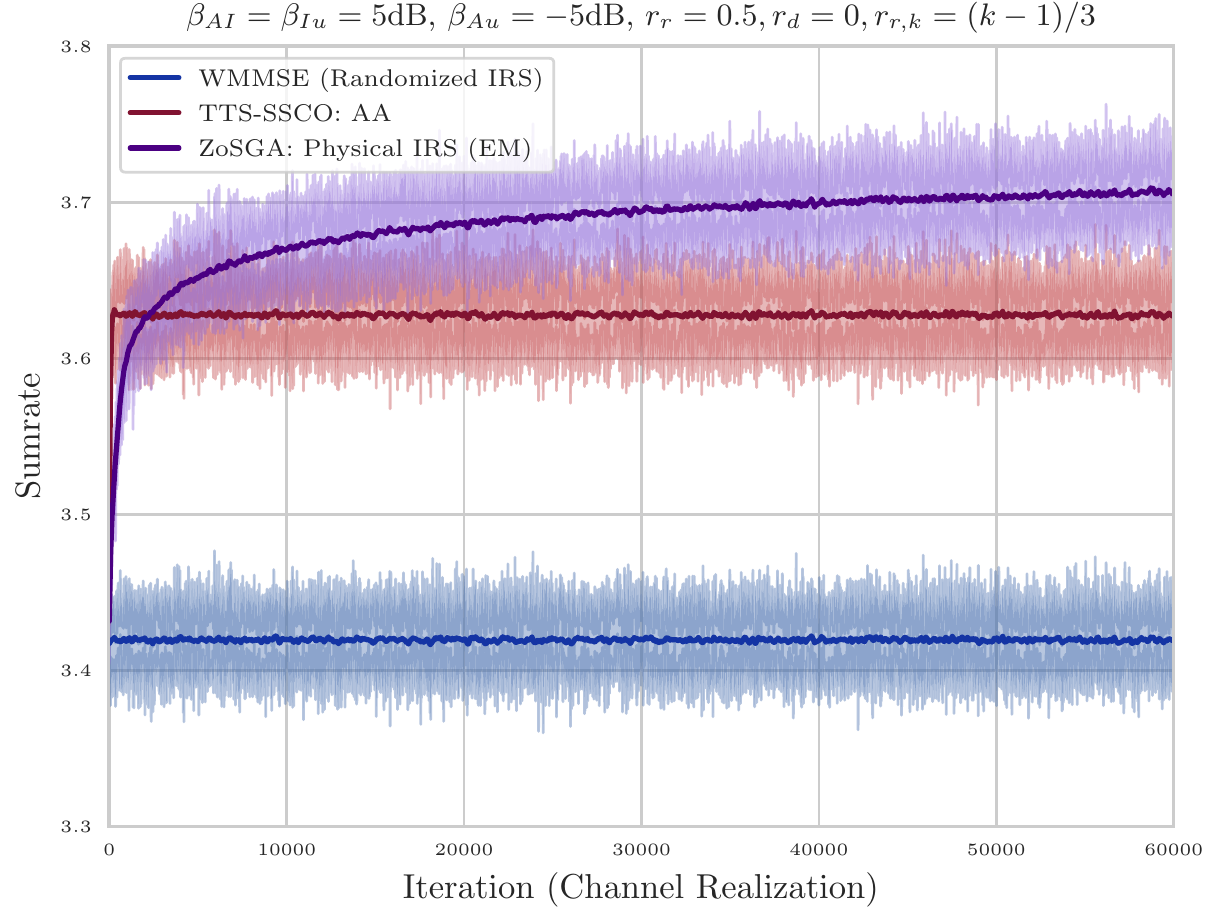}
  \caption{}
  \label{fig:converge1_b}
  \end{subfigure}
  \begin{subfigure}[b]{0.328\textwidth}
  \centering
  \includegraphics[width=\textwidth, height=0.75692307692\textwidth]{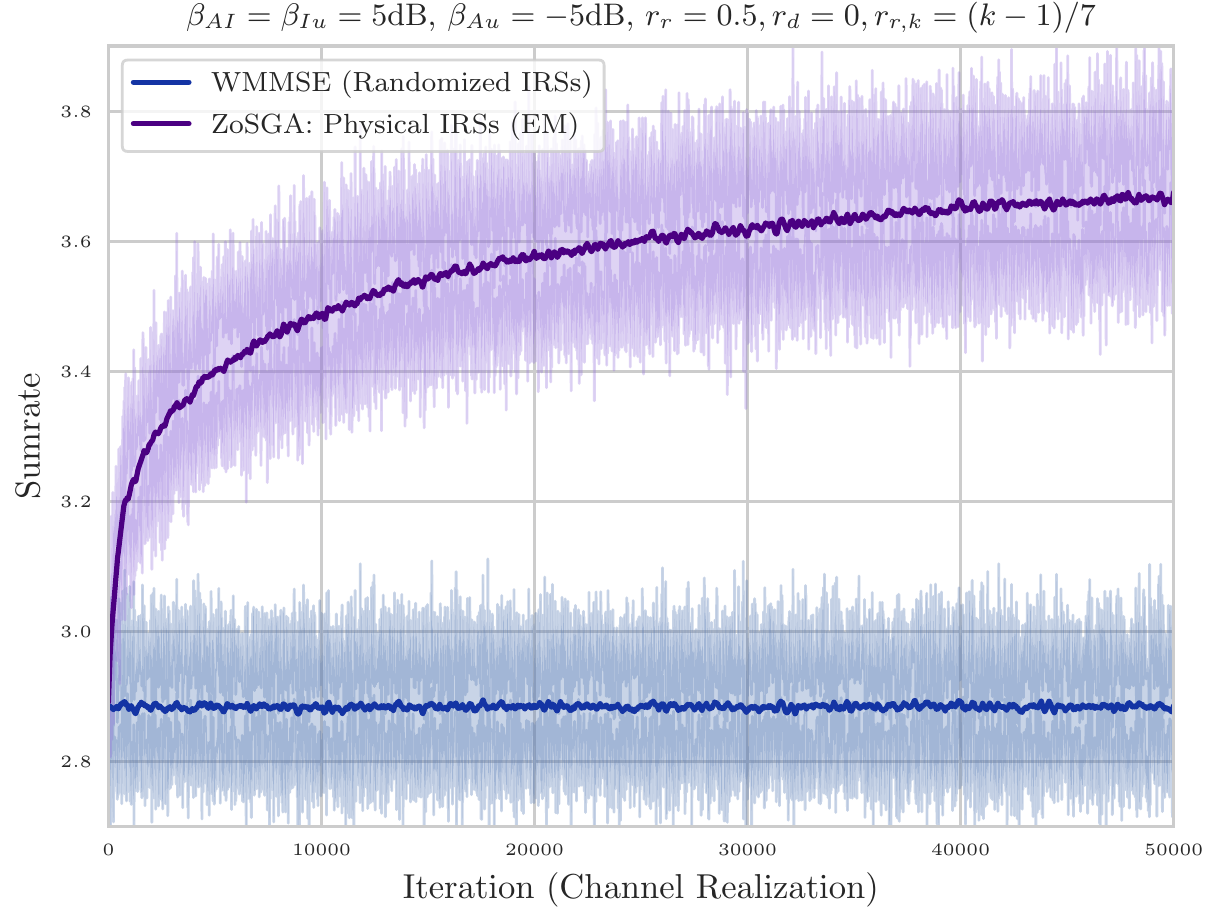}
  \caption{}
  \label{fig:converge3}
  \end{subfigure}
  \begin{subfigure}[b]{0.328\textwidth}
  \centering
  \includegraphics[width=\textwidth, height=0.75692307692\textwidth]{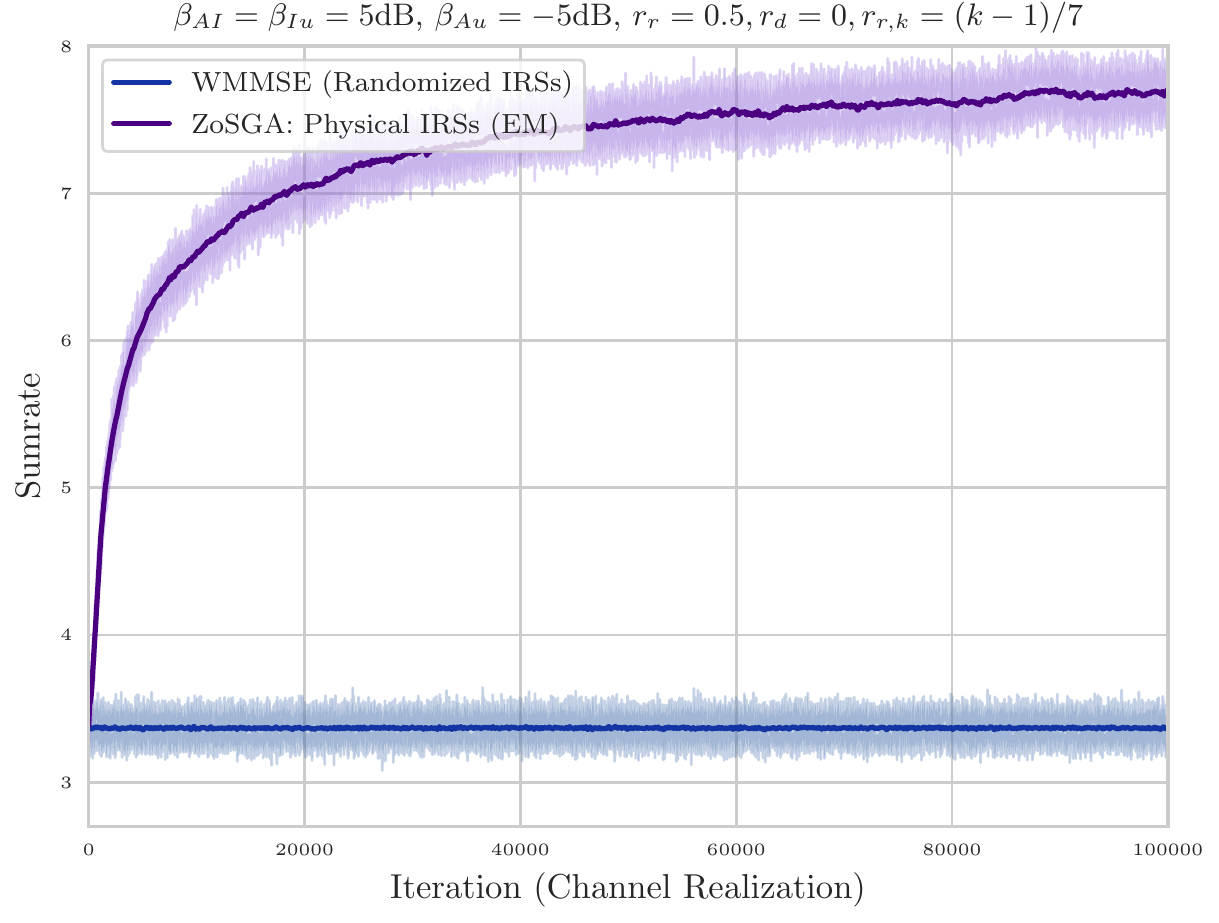}
  \caption{}
  \label{fig:converge3_400}
  \end{subfigure} 
\caption{(a) Average sumrates achieved by WMMSE (random IRS), TTS-SSCO with ideal IRS, and ZoSGA with a physical EM IRS model (network in Figure \ref{fig:env_setup}); Average sumrates achieved by WMMSE (random IRS), and ZoSGA, with four physical (EM model) IRSs with (b) 40 and (c) 400 phase-shift elements, respectively (network in Figure \ref{fig:env_setup_2}).}
\label{fig:main}
\vspace{-10bp}
\end{figure*}

\par The other two impedances, i.e. $Z_d$ and $Z_{patch}$, depend on the properties of the substrate, the dielectric, and the angle of incidence. A detailed description of these is provided in \cite[Section 4]{irs_model:costa2021electromagnetic}. Once all the considered impendances of a patch have been evaluated, its phase-shift coefficient $\theta(C_{var})$ is given by
\begin{equation*}
    \theta(C_{var})= a(C_{var}) e ^{j\phi(C_{var})} = \frac{Z_v(C_{var}) - \zeta_0}{Z_v(C_{var}) + \zeta_0},
\end{equation*}
\noindent where $\zeta_0$ is the free space impedance. Thus, by changing the varactor capacitance $C_{var}$, we can change the varactor impedance $Z_{var}$, which changes the surface impedance $Z_{surf}$, which, in turn, changes the input impedance of an IRS element $Z_v$, finally causing a change in the value of the phase-shift coefficient $\theta(C_{var})$. Due to our assumption that the Floquet theorem holds, this dependency flow is identical for all elements of an IRS. Thus, we may replicate this relation for all IRS elements, say $q$ in number, and define a vector function $\bm{\theta}(\cdot)$ of varactor capacitances $\bm{c}_{var} =  [C^{1}_{var} \, C^2_{var} \, C^{3}_{var} \, \cdots \, C^q_{var}]^\top$ as $\bm{\theta}(\bm{c}_{var}) = [ \theta(C^{1}_{var}) \, \theta(C^2_{var}) \, \theta(C^{3}_{var}) \, \cdots \, \theta(C^q_{var})]^\top$.


\par Now that we have a model for a practical IRS, we would like to evaluate our first wireless network setting, with one IRS, to compare the relative performance drop, if any. As shown in Figure \ref{fig:converge1_b}, the performance of the ZoSGA does drop when constrained in terms of the physical IRS model; this is very much expected, due to a decreased number of degrees of freedom in tuning the IRS parameters. However, the performance gain relative to the random phase-shifts is still substantial. More interestingly, ZoSGA outperforms TTS-SSCO --the latter optimizing both amplitudes and phases in an unconstrained manner-- even in the presence of appreciable practical IRS constraints.

\par To complete our empirical study, which supports our claim of enabling totally model-free optimization of the IRS amplitudes and phase-shifts, we lastly consider an elaborate wireless network setting, as shown in Figure \ref{fig:env_setup_2}, consisting of two APs and four IRSs serving a total of eight users. We consider a MISO downlink scenario where a both APs transmit a common symbol to each user, i.e., the two APs are different only in their position in space. We consider the same channel model and network environment parameters as given in Subsections \ref{subsec:ch_model} and \ref{subsec:Ideal IRS}, respectively. Then, the effective received channel by user $k$ is expressed as 
\begin{equation*} \label{eq: sim_model_2}
\begin{split}
   & \bm{h}_k(\bm{C}_{var},{\omega}) =
    \begin{bmatrix}
        \sum_{i=1}^{2} \underbrace{\bm{G}^\hermtr_{i} \text{Diag}(\bm{\theta}(\bm{c}^i_{var})) \bm{h}^{i}_{r,k}}_{\bm{\theta}_i\text{-non-LoS link}} {+}
        \underbrace{\bm{h}_{d,k,1}}_{\text{AP1 LoS link}} \\
        \sum_{j=3}^{4} \underbrace{\bm{G}^\hermtr_{i} \text{Diag}(\bm{\theta}(\bm{c}^i_{var})) \bm{h}^{i}_{r,k}}_{\bm{\theta}_i\text{-non-LoS link}} {+}
        \underbrace{\bm{h}_{d,k,2}}_{\text{AP2 LoS link}}
    \end{bmatrix},
    \end{split}
\end{equation*}
where, again, $\bm{h}^i_{r,k} = L_{\alpha^i_{Iu}}(d_{Iu,k,i})\breve{\bm{h}}^i_{r,k}$, $\bm{G}^i = L_{\alpha^i_{AI}}(d_{AI,i})\breve{\bm{G}}^i$, and $\bm{h}_{d,k,j} = L_{\alpha^j_{Au}}(d_{Au,k,j})\breve{\bm{h}}_{d,k,j}$ for $i \in \mathbb{N}^+_4$ and $j \in \{1,2\}$. Here, the varactor capacitances of the $i$-th IRS are denoted by the vector $\bm{c}^{i}_{var}$, and where the matrix $\bm{C}_{var} = [ \bm{c}^{1}_{var} \, \bm{c}^{2}_{var} \, \bm{c}^{3}_{var} \, \bm{c}^{4}_{var}]$ combines all the varactor capacitances of the four IRSs. Using the same learning rate scheme as above and a smoothing parameter $\mu=10^{-12}$, we optimize the system sumrate using ZoSGA. 
\par We averaged the results of {40} different simulations in Figure \ref{fig:converge3} to show not only the performance gain due to IRS capacitance tuning, but also the robustness of the approach under different realizations of the channels. The fact that the proposed approach is able to optimize a complicated network such as that in Figure \ref{fig:env_setup_2}, without any model information, verifies our claim of true model-free optimization capability of ZoSGA. We conjecture that the proposed optimization scheme can tackle a wide-range of problems arising in practical applications, with little to no additional input from the user.


\section{Conclusions} \label{sec: conclusions}
In this paper we introduced a zeroth-order stochastic gradient ascent (ZoSGA) method for the solution of two-stage stochastic programs with applications to model-free optimal beamforming for passive IRS-assisted stochastic network utility maximization. ZoSGA is amenable to rigorous convergence analysis and achieves state-of-the-art convergence rate under very general assumptions, capturing a wide range of realistic scenarios. By specializing to the case of sumrate maximization, we numerically demonstrated that ZoSGA outperforms current state-of-the-art model-based methodologies on three distinct network settings, yielding solutions of substantially higher quality and in a computationally efficient manner, while evading practical limitations that are inherent in current methods. Our numerical results confirmed that ZoSGA learns (near-)optimal passive IRS beamformers based solely on conventional effective CSI and in the absence of channel models and spatial network configuration information, also verifying our theoretical findings.

\section*{Appendices}
\addcontentsline{toc}{section}{Appendices}
\renewcommand{\thesubsection}{\Alph{subsection}}

\subsection{Wirtinger Gradient Derivation} \label{apdx: Wirtinger gradient}
\par In light of conditions \textbf{(A1)}--\textbf{(A3)} of Assumption \ref{assumption: two-stage problem}, we can easily show that $F(\bm{W},\bm{H}(\bm{\theta},\omega))$ (i.e., the objective function in \eqref{eqn: second-stage problem}) admits an explicit usual (Fr\'echet) gradient, for all $\bm{\theta} \in \mathcal{U}$ and a.e. $\omega \in \Omega$. We do this by utilizing elements of \emph{Wirtinger calculus} (see \cite[Section 4]{arXiv:Kreutz-Delgado} for a detailed exposition).
\par Indeed, in order to evaluate the gradient of $F\left(\bm{W},\bm{H}(\cdot,\omega)\right)$, i.e., $\nabla_{\bm{\theta}} F\left(\bm{W},\bm{H}(\cdot,\omega)\right) \colon \mathcal{U} \rightarrow \mathbb{R}^S$ , we consider its \emph{Wirtinger cogradient} (a row vector; see \cite[Section 4.2]{arXiv:Kreutz-Delgado}), defined as
\begin{equation*}
    \begin{split}
 & \frac{\partial^{\circ}}{\partial \bm{z}} F\left(\bm{W},\bm{H}(\bm{z},\omega)\right)  \triangleq \frac{1}{2}\bigg( \frac{\partial}{\partial \Re(\bm{z})}F\left(\bm{W},\bm{H}(\bm{z},\omega)\right) \\&\qquad\qquad\qquad\qquad\qquad\quad - j \frac{\partial}{\partial \Im(\bm{z})}F\left(\bm{W},\bm{H}(\bm{z},\omega)\right)\bigg), 
\end{split}
\end{equation*}
\noindent noting that $\bm{H}(\bm{z},\omega)$ is constant relative to $\Im(\bm{z})$, i.e.
\[\bm{H}(\bm{z},\omega) = \bm{H}(\bm{x}+j\bm{y},\omega) = \bm{H}(\bm{x},\omega), \quad \forall (\bm{x},\bm{y}) \in \mathcal{U}\times\mathcal{U}, \]
\noindent and hence so is $F\left(\bm{W},\bm{H}(\bm{z},\omega)\right).$ It then follows that
\[ \frac{\partial^{\circ}}{\partial \bm{z}} F\left(\bm{W},\bm{H}(\bm{z},\omega)\right) = \frac{1}{2} \left( \nabla_{\bm{x}} F\left(\bm{W},\bm{H}(\bm{x},\omega)\right) \right)^\top. \]
\noindent Using the Wirtinger chain rule (see \cite[Eq. (32)]{arXiv:Kreutz-Delgado}), we obtain
\begin{equation*}
    \begin{split}
& \frac{\partial^{\circ}}{\partial \bm{z}} F\left(\bm{W},\bm{H}(\bm{z},\omega)\right) = \frac{\partial^{\circ}}{\partial \bm{z}} F\left(\bm{W},\bm{z}\right) \bigg\vert_{\bm{z} = \bm{H}(\bm{x},\omega)} \frac{\partial^{\circ}}{\partial \bm{z}} \bm{H}(\bm{z},\omega) \\ &\qquad\qquad\qquad\qquad\,\,\,\,\,+ \frac{\partial^{\circ}}{\partial \overline{\bm{z}}} F\left(\bm{W},\bm{z}\right) \bigg\vert_{\bm{z} = \bm{H}(\bm{x},\omega)} \frac{\partial^{\circ}}{\partial \bm{z}} \overline{\bm{H}(\bm{z},\omega)},
\end{split}
\end{equation*}
\noindent where $\overline{\bm{z}}$ denotes the complex conjugate of $\bm{z}$. It follows that
\begin{equation*}
\begin{split}
    &\hspace{-12pt}\frac{\partial^{\circ}}{\partial \bm{z}} \bm{H}(\bm{z},\omega) \\
    & = \frac{1}{2}\left(\frac{\partial}{\partial\Re(\bm{z})} \bm{H}(\bm{z},\omega) - j \frac{\partial}{\partial \Im(\bm{z})} \bm{H}(\bm{z},\omega) \right)\\
  &  = \frac{1}{2} \left( \left( \nabla_{\bm{x}} \Re\left(\bm{H}(\bm{x},\omega) \right)\right)^\top + j\left(\nabla_{\bm{x}} \Im\left(\bm{H}(\bm{x},\omega) \right)\right)^\top\right),
\end{split}
\end{equation*}
\noindent and 
\[ \frac{\partial^{\circ}}{\partial \bm{z}} \overline{\bm{H}(\bm{z},\omega)} = \overline{\left(\frac{\partial^{\circ}}{\partial \bm{z}} \bm{H}(\bm{z},\omega) \right)}.\]
\noindent From the real-valuedness of $F(\bm{W},\bm{H}(\cdot,\omega))$, we have
\begin{equation*}
    \begin{split}
      & \hspace{-12pt} \frac{\partial^{\circ}}{\partial \bm{z}} F\left(\bm{W},\bm{H}(\bm{x},\omega)\right) 
      \\ &= \frac{\partial^{\circ}}{\partial \bm{z}} F\left(\bm{W},\bm{z}\right) \bigg\vert_{\bm{z} = \bm{H}(\bm{x},\omega)} \frac{\partial^{\circ}}{\partial \bm{z}} \bm{H}(\bm{z},\omega)\\ 
      & \quad\quad + \overline{\left(\frac{\partial^{\circ}}{\partial \bm{z}} F\left(\bm{W},\bm{z}\right) \bigg\vert_{\bm{z} = \bm{H}(\bm{x},\omega)} \frac{\partial^{\circ}}{\partial \bm{z}} \bm{H}(\bm{z},\omega)\right)} \\
    & = 2\Re\left(\frac{\partial^{\circ}}{\partial \bm{z}}F(\bm{W},\bm{z})\bigg\vert_{\bm{z} = \bm{H}(\bm{x},\omega)} \frac{\partial^{\circ}}{\partial \bm{z}} \bm{H}(\bm{z},\omega) \right).
    \end{split}
\end{equation*}
\noindent Thus, for any $\bm{\theta} \in \mathcal{U}$ and a.e. $\omega \in \Omega$, we obtain
\begin{equation*} 
    \begin{split}
& \nabla_{\bm{\theta}} F\left(\bm{W},\bm{H}(\bm{\theta},\omega)\right)  \\ &\quad =\ 2 \nabla_{\bm{\theta}} \Re\left(\bm{H}(\bm{\theta},\omega)\right)\Re\left(\frac{\partial^{\circ}}{\partial \bm{z}} F(\bm{W},\bm{z}) \bigg\vert_{\bm{z} = \bm{H}(\bm{\theta},\omega)} \right)^\top \\  & \qquad + 2 \nabla_{\bm{\theta}} \Im\left(\bm{H}(\bm{\theta},\omega)\right)\Re\left(j\frac{\partial^{\circ}}{\partial \bm{z}} F(\bm{W},\bm{z}) \bigg\vert_{\bm{z} = \bm{H}(\bm{\theta},\omega)} \right)^\top. 
\end{split}
\end{equation*}
\vspace{-13pt}
\subsection{Weak Concavity of the First-Stage Objective Function} \label{apdx: Weak concavity}
\par In what follows, we present three typical situations under which condition \textbf{(A4)} of Assumption \ref{assumption: two-stage problem} is satisfied for problem \eqref{eqn: two-stage problem}. To simplify the discussion, we assume that $\mathcal{W} = \{\bm{W} \colon \|\bm{W}\|^2 \leq P\}$, which is the constraint utilized in \eqref{eqn: sumrate two-stage problem}, noting that this is done without loss of generality. 
\par \textbf{Cases 1 and 2 - Strong second-order sufficient optimality:} The first two cases rely on the strong second-order sufficient optimality conditions for the second-stage problem \eqref{eqn: second-stage problem}. In this case, for any $\bm{\theta} \in \Theta$ and a.e. $\omega \in \Omega$, the Lagrangian associated with \eqref{eqn: second-stage problem} reads
\[ L(\bm{W},\lambda; \bm{\theta},\omega) = F(\bm{W},\bm{H}(\bm{\theta},\omega)) + \lambda\left(\|\bm{W}\|^2 - P\right),\]
\noindent where the admissible Lagrange multipliers are nonnegative, i.e. $\lambda \geq 0$. Let any $\bm{W}^*(\bm{\theta},\omega) \in \arg\max_{\bm{W} \in \mathcal{W}} F(\bm{W},\bm{H}(\bm{\theta},\omega))$. Then, the strong second-order sufficient optimality conditions require \emph{strict complementarity slackness} (i.e. $\lambda^* > 0$ if $\|\bm{W}^*(\bm{\theta},\bm{H}(\bm{\theta},\omega))\|^2 = P$, and $\lambda^* = 0$ otherwise, where $\lambda^*$ is an optimal Lagrange multiplier associated to $\bm{W}^*(\bm{\theta},\bm{H}(\bm{\theta},\omega))$), as well as that $\nabla_{\bm{W}}^2 L(\bm{W},\lambda; \bm{\theta},\omega)\vert_{(\bm{W},\lambda) = \left(\bm{W}^*(\bm{\theta},\omega),\lambda^*\right)}$ is nonsingular. We note that in the case of problem \eqref{eqn: sumrate two-stage problem} (i.e. the sumrate example), strict complementarity slackness can be shown without any additional assumptions. 
\par In the first case (\textbf{Case 1}), we assume that the second-stage problem \eqref{eqn: second-stage problem} admits a unique solution for each $\bm{\theta} \in \Theta$ and a.e. $\omega \in \Omega$. It then follows that the mapping $\bm{\theta} \mapsto \max_{\bm{W} \in \mathcal{W}} F(\bm{W},\bm{H}(\bm{\theta},\omega))$ is twice continuously differentiable on ${\Theta}$ by utilizing \cite[Lemma 2.2]{MathOR:Shapiro} (which, in turn, utilizes the Implicit Function Theorem, e.g. see \cite[Theorem 1B.1]{Springer:DonRock}). Twice continuously differentiable functions on a compact set (in this case ${\Theta}$) are, in fact, Lipschitz smooth on that set. However, Lipschitz smooth functions are both weakly convex and weakly concave (see \cite[Proposition 4.12]{Vial_WeakConvexity}), and thus we are done.
\par In the second case (\textbf{Case 2}), instead of assuming that the solution set of the second-stage problem is a singleton, we assume that $\bm{H}(\cdot,\omega)$ is real analytic (which, for example, is true in Section \ref{sec: Simulations}). Then, assuming that $F(\cdot,\cdot)$ is the sumrate (as in \eqref{eqn: sumrate two-stage problem}), it follows that the mapping $\bm{\theta} \mapsto \max_{\bm{W} \in \mathcal{W}} F(\bm{W},\bm{H}(\bm{\theta},\omega))$  is \emph{sub-analytic} on an open bounded set $\mathcal{U}$ (where $\mathcal{U}$ is given in condition \textbf{(A3)} of Assumption \ref{assumption: two-stage problem}; for a proof of this result, see \cite[Example 4]{SIAMOpt:Daniilidis}). In turn, this implies (e.g. see \cite[Theorem 2.3]{denkowska2018upc}) that the function $\max_{\bm{W} \in \mathcal{W}} F(\bm{W},\bm{H}(\bm{\theta},\omega)) - F(\bm{W},\bm{H}(\bm{\theta},\omega))$ satisfies the \L{}ojasiewicz inequality with uniform exponent, i.e. for a.e. $\omega \in \Omega$ and for each $\bm{\theta} \in \mathcal{U}$, there exists a positive constant $\eta$ and a positive subanalytic function $C(\bm{\theta})$, such that for every $(\bm{\theta},\bm{W}) \in \mathcal{U} \times \mathcal{W}$, we have
    \begin{equation*}
    \begin{split}
        &\textnormal{dist}\left(\bm{W}, \arg\max_{\bm{W} \in \mathcal{W}} F\left(\bm{W},\bm{H}(\bm{\theta},\omega)\right)\right) \\ & \quad\leq C(\bm{\theta}) \left(\max_{\bm{W} \in \mathcal{W}} F(\bm{W},\bm{H}(\bm{\theta},\omega)) - F(\bm{W},\bm{H}(\bm{\theta},\omega))\right)^{\eta}. 
    \end{split}
    \end{equation*}
\noindent Note that we can always find a compact set $\Theta''$ and an open set $\Theta'$ such that $\mathcal{U} \supset \Theta'' \supset \Theta' \supset \Theta$. In what follows we make the reasonable assumption that $C(\bm{\theta})$ is uniformly bounded on the compact set $\Theta''$. 
\begin{lemma}  \label{lemma: compact-valuedness of argmax}
Let $\Theta'' \supset \Theta' \supset \Theta$, be a compact set, with $\Theta'$ some open set. Given conditions \textnormal{\textbf{(A1)}}--\textnormal{\textbf{(A3)}} of Assumption \textnormal{\ref{assumption: two-stage problem}}, the multifunction $\bm{\theta} \mapsto \arg\max_{\bm{W} \in \mathcal{W}} F(\bm{W},\bm{H}(\bm{\theta},\omega))$ is nonempty and compact-valued on $\Theta''$, for a.e. $\omega \in \Omega$.
\end{lemma}
\begin{proof}
\par We have that $F(\cdot,\cdot)$ is (real) continuously differentiable, and that for any $\bm{\theta} \in \Theta''$ and a.e. $\omega \in \Omega$, $\arg\max_{\bm{W} \in \mathcal{W}}  F(\bm{W},\bm{H}(\bm{\theta},\omega))$ is non-empty. Additionally, both $\Theta''$ and $\mathcal{W}$ are assumed to be compact, and $\bm{H}(\cdot,\omega)$ is continuously differentiable and thus has compact range. In turn, we obtain that $F(\cdot,H(\cdot,\omega))$ also has compact range on $\Theta''\times \mathcal{W}$ and jointly continuous. We complete the proof by applying Berge's \emph{maximum theorem} (see \cite{PrincUP:EfeOk}).
\end{proof}
\begin{lemma}\label{lemma: local continuity of the selection}
Let conditions \textnormal{\textbf{(A1)}--\textbf{(A3)}} of Assumption \textnormal{\ref{assumption: two-stage problem}} hold, along with the aforementioned conditions of \textbf{Case 2}. Then, for any $\bm{\theta} \in \Theta$ and a.e. $\omega \in \Omega$, and each selection 
\[\bm{W}^*(\bm{\theta},\omega) \in \arg\max_{\bm{W} \in \mathcal{W}} F(\bm{W},\bm{H}(\bm{\theta},\omega)),\]
\noindent there exists a sequence of selections 
\[\left\{\widetilde{\bm{W}}^*(\bm{\theta}+\bm{z}_k,\omega)\right\}_{k=0}^{\infty},\] 
\noindent where $\widetilde{\bm{W}}^*(\bm{\theta}+\bm{z}_k,\omega)\in \arg\max_{\bm{W} \in \mathcal{W}} F(\bm{W},\bm{H}(\bm{\theta}+\bm{z}_k,\omega))$, for some sequence $\bm{\theta} + \{\bm{z}_k\}_{k = 0}^{\infty} \subset \Theta'$, such that
\[ \lim_{\|\bm{z}_k\|\rightarrow 0} \widetilde{\bm{W}}^*(\bm{\theta}+\bm{z}_k,\omega) = \bm{W}^*(\bm{\theta},\omega).\]
\end{lemma}
\begin{proof}
\par We fix some $(\bm{\theta},\omega) \in \Theta \times \Omega$ and a bounded open set $\Theta' \supset \Theta$. From our assumptions, we know that there must exist two positive constants $C,\ \eta$, such that for every $(\bm{W},\bm{\theta}+\bm{z}) \in \mathcal{W}\times \Theta'$
\begin{equation*}
\begin{split}
 &\textnormal{dist}\left(\bm{W},\arg\max_{\bm{W} \in \mathcal{W}} F(\bm{W},\bm{H}(\bm{\theta}+\bm{z},\omega))\right) \\
 &\quad \leq C\left(\max_{\bm{W} \in \mathcal{W}}  F(\bm{W},\bm{H}(\bm{\theta}+\bm{z},\omega))-F(\bm{W},\bm{H}(\bm{\theta}+\bm{z},\omega))\right)^{\eta}.
 \end{split}
 \end{equation*}
\noindent From Lemma \ref{lemma: compact-valuedness of argmax} (in particular, from the closed-valuedness of $\arg\max_{\bm{W} \in \mathcal{W}} F(\bm{W},\bm{H}(\bm{\theta},\omega))$ on the compact set $\Theta'' \supset \Theta'$), for every $\bm{W} \in \mathcal{W}$ and any $\bm{z}$ such that $\bm{\theta}+\bm{z} \in \Theta'$, there exists a selection $\widetilde{\bm{W}}^*(\bm{\theta}+\bm{z},\omega)$ such that
\begin{equation*}
\begin{split}
&\left\|\widetilde{\bm{W}}^*(\bm{\theta}+\bm{z},\omega) -  \bm{W}\right\| \\
&\quad = \textnormal{dist}\left(\bm{W},\arg\max_{\bm{W} \in \mathcal{W}} F(\bm{W},\bm{H}(\bm{\theta}+\bm{z},\omega))\right). 
\end{split}
\end{equation*}
\noindent Continuity of $\max_{\bm{W} \in \mathcal{W}} F(\bm{W},\bm{H}(\cdot,\omega)) - F(\bm{W},\bm{H}(\cdot,\omega))$ then yields the desired result, since we can consider a sequence $\bm{\theta}+\{\bm{z}_k\}_{k=0}^{\infty} \subset \Theta'$ such that $\|\bm{z}_k\|\rightarrow 0$.
\end{proof}
\begin{lemma}\label{lemma: differentiability of sample objective based on Lojasiewicz}
Let conditions \textnormal{\textbf{(A1)}--\textbf{(A3)}} of Assumption \textnormal{\ref{assumption: two-stage problem}} hold, along with the aforementioned conditions of \textbf{Case 2}. For a.e. $\omega \in \Omega$, the function $\max_{\bm{W} \in \mathcal{W}}F(\bm{W},\bm{H}(\cdot,\omega))$ is Fréchet differentiable on $\Theta$, irrespectively of the selection $\bm{W}^*(\cdot,\omega)$ from $\arg\max_{\bm{W} \in \mathcal{W}} F(\bm{W},\bm{H}(\cdot,\omega))$.
\end{lemma}
\begin{proof}
   \par Let us fix some $(\bm{\theta},\omega) \in \Theta \times \Omega$ and some selection $\bm{W}^*(\bm{\theta},\omega) \in \arg\max_{\bm{W} \in \mathcal{W}} F(\bm{W},\bm{H}(\bm{\theta},\omega))$. Firstly, we note that the function $\max_{\bm{W} \in \mathcal{W}} F(\bm{W},\bm{H}(\bm{\theta},\omega))$ is defined and bounded on an open set $\Theta' \supset \Theta$. Additionally, we have that $F(\cdot,\bm{H}(\bm{\theta},\omega))$ is real-analytic on $\mathcal{W}$ for any $\bm{\theta} \in \Theta'$, and $F(\bm{W},\bm{H}(\cdot,\omega))$ is continuously differentiable on $\Theta'$. By definition, we obtain that
   \begin{equation*}
   \begin{split}
 & \max_{\bm{W} \in \mathcal{W}} F(\bm{W},\bm{H}(\bm{\theta}+\bm{z},\omega)) - \max_{\bm{W} \in \mathcal{W}}F(\bm{W},\bm{H}(\bm{\theta},\omega))\\ 
  &\quad \leq  F\left(\widetilde{\bm{W}}^*(\bm{\theta}+\bm{z},\omega),\bm{H}(\bm{\theta}+\bm{z},\omega)\right) \\
  &\qquad - F\left(\widetilde{\bm{W}}^*(\bm{\theta}+\bm{z},\omega),\bm{H}(\bm{\theta},\omega)\right),
   \end{split}
   \end{equation*}
\noindent for all $\bm{z}$ such that $\bm{\theta}+\bm{z} \in \Theta'$, where $\widetilde{\bm{W}}^*(\bm{\theta}+\bm{z},\omega)$ can be chosen as in Lemma \ref{lemma: local continuity of the selection}. Furthermore, we observe that
\begin{equation*}
\begin{split}
& \max_{\bm{W} \in \mathcal{W}} F(\bm{W},\bm{H}(\bm{\theta}+\bm{z},\omega))  - \max_{\bm{W} \in \mathcal{W}} F(\bm{W},\bm{H}(\bm{\theta},\omega)) \\
&\quad \geq F(\bm{W}^*(\bm{\theta},\omega),\bm{H}(\bm{\theta}+\bm{z},\omega)) - F(\bm{W}^*(\bm{\theta},\omega),\bm{H}(\bm{\theta},\omega)).
\end{split}
\end{equation*}
\noindent In other words, we have
\begin{equation*}
\begin{split}
\label{eqn: lemma: differentiability of G, double inequality}
&  F(\bm{W}^*(\bm{\theta},\omega),\bm{H}(\bm{\theta}+\bm{z},\omega)) - F(\bm{W}^*(\bm{\theta},\omega),\bm{H}(\bm{\theta},\omega)) \\
& \quad \leq \max_{\bm{W} \in \mathcal{W}} F(\bm{W},\bm{H}(\bm{\theta}+\bm{z},\omega))  - \max_{\bm{W} \in \mathcal{W}} F(\bm{W},\bm{H}(\bm{\theta},\omega)) \\
&\quad \leq  F\left(\widetilde{\bm{W}}^*(\bm{\theta}+\bm{z},\omega),\bm{H}(\bm{\theta}+\bm{z},\omega)\right) \\
  &\qquad - F\left(\widetilde{\bm{W}}^*(\bm{\theta}+\bm{z},\omega),\bm{H}(\bm{\theta},\omega)\right)
\end{split}
\end{equation*}
\par Since $F(\bm{W},\bm{H}(\cdot,\omega))$ is differentiable on $\Theta'$, for any $\bm{W} \in \mathcal{W}$, we can employ the mean value theorem to show that for every $\bm{z}$ satisfying $\bm{\theta}+ \bm{z} \in \Theta'$, there exists $c \equiv c(\bm{\theta}+ \bm{z},\bm{W},\omega) \in (0,1)$ such that
\begin{equation*}
    \begin{split}
       & F\left(\widetilde{\bm{W}}^*(\bm{\theta}+\bm{z},\omega),\bm{H}(\bm{\theta}+\bm{z},\omega)\right)\\
       &\quad - F\left(\widetilde{\bm{W}}^*(\bm{\theta}+\bm{z},\omega),\bm{H}(\bm{\theta},\omega)\right)\\
       &\qquad = \left\langle\nabla_{\bm{\theta}}  F\left(\bm{W},\bm{H}(\bm{\theta}+c\bm{z},\omega)\right),\bm{z}\right\rangle \big\vert_{\bm{W} =\widetilde{\bm{W}}^*(\bm{\theta}+\bm{z},\omega)} .
    \end{split}
\end{equation*}
\noindent Given the previous inequalities, the latter yields
\begin{equation*}
\begin{split}
&F(\bm{W}^*(\bm{\theta},\omega),\bm{H}(\bm{\theta}+\bm{z},\omega)) - F(\bm{W}^*(\bm{\theta},\omega),\bm{H}(\bm{\theta},\omega)) \\
&\qquad -\left\langle\nabla_{\bm{\theta}}  F\left(\bm{W},\bm{H}(\bm{\theta},\omega)\right),\bm{z}\right\rangle\big\vert_{\bm{W} = \bm{W}^*(\bm{\theta},\omega)}\\
&\quad \leq  F\left(\widetilde{\bm{W}}^*(\bm{\theta}+\bm{z},\omega),\bm{H}(\bm{\theta}+\bm{z},\omega)\right) \\
  &\qquad - F\left(\widetilde{\bm{W}}^*(\bm{\theta}+\bm{z},\omega),\bm{H}(\bm{\theta},\omega)\right) \\
&\qquad -\left\langle\nabla_{\bm{\theta}}  F\left(\bm{W},\bm{H}(\bm{\theta},\omega)\right),\bm{z}\right\rangle\big\vert_{\bm{W} = \bm{W}^*(\bm{\theta},\omega)}\\
&\quad \leq \bigg\|\nabla_{\bm{\theta}}  F\left(\bm{W},\bm{H}(\bm{\theta},\omega)\right)\big\vert_{\bm{W} = \bm{W}^*(\bm{\theta},\omega)} \\
&\qquad \quad - \nabla_{\bm{\theta}}  F\left(\bm{W},\bm{H}(\bm{\theta} + c\bm{z},\omega)\right)\big\vert_{\bm{W} =\widetilde{\bm{W}}^*(\bm{\theta}+\bm{z},\omega)}\bigg\| \|\bm{z}\|.\\
\end{split}
\end{equation*}
\noindent Utilizing again the mean value theorem, we also have that there exists $c' \equiv c'(\bm{\theta}+\bm{z},\bm{W},\omega) \in (0,1)$ such that
\begin{equation*}
    \begin{split}
        & F\left({\bm{W}}^*(\bm{\theta},\omega),\bm{H}(\bm{\theta}+\bm{z},\omega)\right)- F\left({\bm{W}}^*(\bm{\theta},\omega),\bm{H}(\bm{\theta},\omega)\right)\\
       &\qquad = \left\langle\nabla_{\bm{\theta}}  F\left(\bm{W},\bm{H}(\bm{\theta}+c'\bm{z},\omega)\right),\bm{z}\right\rangle \big\vert_{\bm{W} ={\bm{W}}^*(\bm{\theta},\omega)} .       
    \end{split}
\end{equation*}
\noindent As before, this implies that 
\begin{equation*}
\begin{split}
&F(\bm{W}^*(\bm{\theta},\omega),\bm{H}(\bm{\theta}+\bm{z},\omega)) - F(\bm{W}^*(\bm{\theta},\omega),\bm{H}(\bm{\theta},\omega)) \\
&\qquad -\left\langle\nabla_{\bm{\theta}}  F\left(\bm{W},\bm{H}(\bm{\theta},\omega)\right),\bm{z}\right\rangle\big\vert_{\bm{W} = \bm{W}^*(\bm{\theta},\omega)}\\
&\quad \geq -\bigg\|\nabla_{\bm{\theta}}  F\left(\bm{W},\bm{H}(\bm{\theta} + c'\bm{z},\omega)\right)\big\vert_{\bm{W} = \bm{W}^*(\bm{\theta},\omega)} \\
&\qquad \quad - \nabla_{\bm{\theta}}  F\left(\bm{W},\bm{H}(\bm{\theta},\omega)\right)\big\vert_{\bm{W} ={\bm{W}}^*(\bm{\theta},\omega)}\bigg\| \|\bm{z}\|.\\
\end{split}
\end{equation*}
\noindent Hence, we have shown that 
\begin{equation*}
\begin{split}
&-\bigg\|\nabla_{\bm{\theta}}  F\left(\bm{W},\bm{H}(\bm{\theta} + c'\bm{z},\omega)\right)\big\vert_{\bm{W} = \bm{W}^*(\bm{\theta},\omega)} \\
&\qquad \quad - \nabla_{\bm{\theta}}  F\left(\bm{W},\bm{H}(\bm{\theta},\omega)\right)\big\vert_{\bm{W} ={\bm{W}}^*(\bm{\theta},\omega)}\bigg\| \\
&\quad \leq \frac{1}{\|\bm{z}\|}\bigg( F(\bm{W}^*(\bm{\theta},\omega),\bm{H}(\bm{\theta}+\bm{z},\omega))\\
&\qquad - F(\bm{W}^*(\bm{\theta},\omega),\bm{H}(\bm{\theta},\omega)) \\
&\qquad -\left\langle\nabla_{\bm{\theta}}  F\left(\bm{W},\bm{H}(\bm{\theta},\omega)\right),\bm{z}\right\rangle\big\vert_{\bm{W} = \bm{W}^*(\bm{\theta},\omega)}\bigg)\\
&\quad \leq \bigg\|\nabla_{\bm{\theta}}  F\left(\bm{W},\bm{H}(\bm{\theta},\omega)\right)\big\vert_{\bm{W} = \bm{W}^*(\bm{\theta},\omega)} \\
&\qquad \quad - \nabla_{\bm{\theta}}  F\left(\bm{W},\bm{H}(\bm{\theta} + c\bm{z},\omega)\right)\big\vert_{\bm{W} =\widetilde{\bm{W}}^*(\bm{\theta}+\bm{z},\omega)}\bigg\|
\end{split}
\end{equation*}
\noindent Thus, by utilizing Lemma \ref{lemma: local continuity of the selection}, we have shown that
\begin{equation*}
    \begin{split}
        & \lim_{\|\bm{z}\| \rightarrow 0} \frac{1}{\|\bm{z}\|}\bigg(F(\bm{W}^*(\bm{\theta},\omega),\bm{H}(\bm{\theta}+\bm{z},\omega)) \\
        &\qquad \quad  - F(\bm{W}^*(\bm{\theta},\omega),\bm{H}(\bm{\theta},\omega))\\
&\qquad \quad -\left\langle\nabla_{\bm{\theta}}  F\left(\bm{W},\bm{H}(\bm{\theta},\omega)\right),\bm{z}\right\rangle\big\vert_{\bm{W} = \bm{W}^*(\bm{\theta},\omega)}\bigg) = 0,
    \end{split}
\end{equation*}
\noindent which implies that $\max_{\bm{W} \in \mathcal{W}}F(\bm{W},\bm{H}(\cdot,\omega))$ is (Fréchet) differentiable at $\bm{\theta}\in \Theta$, for a.e. $\omega \in \Omega$, and its gradient reads as
\[\nabla_{\bm{\theta}} \max_{\bm{W} \in \mathcal{W}}F(\bm{W},\bm{H}(\bm{\theta},\omega)) = \nabla_{\bm{\theta}} F(\bm{W},\bm{H}(\bm{\theta},\omega))\big\vert_{\bm{W} = \bm{W}^*(\bm{\theta},\omega)}\]
\noindent for any selection $\bm{W}^*(\bm{\theta},\omega) \in \arg\max_{\bm{W} \in \mathcal{W}} F(\bm{W},\bm{H}(\bm{\theta},\omega))$. This completes the proof.
\end{proof}
\begin{theorem} \label{thm: weak convexity based on Lojasiewicz}
    Let conditions \textnormal{\textbf{(A1)}--\textbf{(A3)}} of Assumption \textnormal{\ref{assumption: two-stage problem}} hold, along with the aforementioned conditions of \textbf{Case 2}. For a.e. $\omega \in \Omega$, the function $\max_{\bm{W} \in \mathcal{W}}F(\bm{W},\bm{H}(\cdot,\omega))$ is Lipschitz smooth and thus weakly concave.
\end{theorem}
\begin{proof}
\par Under our assumptions, we know from Lemma \ref{lemma: differentiability of sample objective based on Lojasiewicz} that the function $\max_{\bm{W} \in \mathcal{W}} F(\bm{W},\bm{H}(\bm{\theta},\omega))$ is Frech\'et differentiable on $\Theta$ and its gradient reads as
\[ \nabla_{\bm{\theta}} \max_{\bm{W} \in \mathcal{W}} F(\bm{W},\bm{H}(\bm{\theta},\omega)) = F(\bm{W}^*(\bm{\theta},\omega),\bm{H}(\bm{\theta},\omega)),\]
\noindent for any $\bm{W}^*(\bm{\theta},\omega) \in \arg\max_{\bm{W} \in \mathcal{W}} F(\bm{W},\bm{H}(\bm{\theta},\omega))$. At the same time, from the strong second-order sufficient conditions, we can apply the Implicit Function Theorem (as in \cite[Lemma 2.1]{MathOR:Shapiro}; see also \cite[Theorem 1B.1]{Springer:DonRock}) to the Lagrangian of the second-stage problem, which implies that for every $\bm{\theta}_{\circ} \in \Theta$, and any arbitrary selection $\bm{W}^*(\bm{\theta}_{\circ},\omega) \in \arg\max_{\bm{W}\in \mathcal{W}} F(\bm{W},\bm{H}(\bm{\theta}_{\circ},\omega))$, there exists a neighbourhood $\Theta_{\circ} \ni \bm{\theta}_{\circ}$, on which there exists a unique continuously differentiable mapping $\bm{W}_{\circ}^*(\cdot,\omega): \bm{\theta} \mapsto \arg \max_{\bm{W} \in \mathcal{W}} F(\bm{W},\bm{H}(\bm{\theta},\omega))$, such that $\bm{W}_{\circ}^*(\bm{\theta}_{\circ},\omega) = \bm{W}^*(\bm{\theta}_{\circ},\omega)$. Let us fix this neighbourhood $\Theta_{\circ}$, and cosnider two arbitrary points $\bm{\theta}_1, \bm{\theta}_2 \in \widetilde{\Theta}_{\circ} \subset \Theta_{\circ}$ (assuming that $\widetilde{\Theta}_{\circ}$ is a compact set), along with some arbitrary selections $\bm{W}^*_1 \triangleq \bm{W}^*(\bm{\theta}_1,\omega)$, $\bm{W}^*_2 \triangleq \bm{W}^*(\bm{\theta}_2,\omega)$.
\par Using the previous points, and letting $\bm{W}^*_{\circ,i} \triangleq \bm{W}^*_{\circ}(\bm{\theta}_i,\omega)$, for $i = 1,2$, we have
\begin{equation*}
    \begin{split}
       & \left\|\nabla_{\bm{\theta}} F(\bm{W}^*_1,\bm{H}(\bm{\theta}_1,\omega)) - \nabla_{\bm{\theta}} F(\bm{W}^*_2,\bm{H}(\bm{\theta}_2,\omega))\right\| \\
              & \quad = \left\|\nabla_{\bm{\theta}} F(\bm{W}^*_{\circ,1},\bm{H}(\bm{\theta}_1,\omega))  - \nabla_{\bm{\theta}} F(\bm{W}^*_{\circ,2},\bm{H}(\bm{\theta}_2,\omega))\right\| \\
       &\quad = \bigg\|\nabla_{\bm{\theta}} F(\bm{W}^*_{\circ,1},\bm{H}(\bm{\theta}_1,\omega)) - \nabla_{\bm{\theta}} F(\bm{W}^*_{\circ,1},\bm{H}(\bm{\theta}_2,\omega)) \\
       &\qquad  + \nabla_{\bm{\theta}} F(\bm{W}^*_{\circ,1},\bm{H}(\bm{\theta}_2,\omega))- \nabla_{\bm{\theta}} F(\bm{W}^*_{\circ,2},\bm{H}(\bm{\theta}_2,\omega))\bigg\|\\
       &\quad \leq \left\|\nabla_{\bm{\theta}} F(\bm{W}^*_{\circ,1},\bm{H}(\bm{\theta}_1,\omega)) - \nabla_{\bm{\theta}} F(\bm{W}^*_{\circ,1},\bm{H}(\bm{\theta}_2,\omega)) \right\| \\
       &\qquad + \left\| \nabla_{\bm{\theta}} F(\bm{W}^*_{\circ,1},\bm{H}(\bm{\theta}_2,\omega))- \nabla_{\bm{\theta}} F(\bm{W}^*_{\circ,2},\bm{H}(\bm{\theta}_2,\omega))\right\|,
    \end{split}
\end{equation*}
\noindent where we used the fact that $\nabla_{\bm{\theta}} \max_{\bm{W} \in \mathcal{W}} F(\bm{W},\bm{H}(\bm{\theta},\omega))$ has the same value independently of the optimum selection, while $\bm{W}^*_{\circ}(\bm{\theta}_i,\omega) \in \arg \max_{\bm{W} \in \mathcal{W}} F(\bm{W},\bm{H}(\bm{\theta}_i,\omega))$, for $i = 1,2$. However, Lipschitz smnoothness of the function $F(\bm{W},\bm{H}(\cdot,\omega))$ implies that there exists some constant $L_1 > 0$ such that
\begin{equation*} 
\begin{split}
&\left\|\nabla_{\bm{\theta}} F(\bm{W}^*_{\circ,1},\bm{H}(\bm{\theta}_1,\omega)) - \nabla_{\bm{\theta}} F(\bm{W}^*_{\circ,1},\bm{H}(\bm{\theta}_2,\omega)) \right\| \\
&\qquad \leq L_1 \|\bm{\theta}_1 - \bm{\theta}_2\|.
\end{split}
\end{equation*}
\noindent Then, we observe that 
\begin{equation*} 
\begin{split}
&\left\| \nabla_{\bm{\theta}} F(\bm{W}^*_{\circ,1},\bm{H}(\bm{\theta}_2,\omega))- \nabla_{\bm{\theta}} F(\bm{W}^*_{\circ,2},\bm{H}(\bm{\theta}_2,\omega))\right\|\\
&\qquad \leq L_2 \|\bm{W}^*_{\circ}(\bm{\theta}_1,\omega) - \bm{W}^*_{\circ}(\bm{\theta}_2,\omega)\| \leq L_2 L_3(\widetilde{\Theta}_{\circ}) \|\bm{\theta}_1- \bm{\theta}_2\|,
\end{split}
\end{equation*}
\noindent where we used the $L_2$-Lipschitz smoothnes of $F(\cdot, \bm{H}(\bm{\theta},\omega))$, and the fact that $\bm{W}^*_{\circ}(\cdot,\omega)$ is continuously differentiable on the compact set $\widetilde{\Theta}_{\circ}$, and thus $L_3(\widetilde{\Theta}_{\circ})$-Lipschitz continuous. In other words, the function $\max_{\bm{W} \in \mathcal{W}} F(\bm{W},\bm{H}(\bm{W},\omega))$ is locally Lipschitz smooth. However, since $\Theta$ is compact, this is equivalent to saying that it is globally Lipschitz smooth, which implies that it is both weakly convex and weakly concave (see \cite[Proposition 4.12]{Vial_WeakConvexity}). This completes the proof.
\end{proof}
\par \textbf{Case 3 - Lack of Hessian Invertibility:} Finally, we should mention that if the solution set of the second-stage problem \eqref{eqn: second-stage problem} is not a singleton, while it also does not satisfy the strong second-order sufficient optimality conditions utilized in \textbf{Case 2}, one can still show that condition \textbf{(A4)} of Assumption \ref{assumption: two-stage problem} holds, by utilizing the analysis of \cite[Section 4]{MathOR:Shapiro}. Indeed, under some regularity conditions, coupled with an assumption of connectedness of the solution set of the second-stage problem, it follows that the function $\max_{\bm{W} \in \mathcal{W}} F(\bm{W},\bm{\Theta}(\cdot,\omega))$ is twice continuously differentiable on $\Theta$, and thus weakly concave. The details are omitted, but the reader is referred to \cite[Section 4]{MathOR:Shapiro} for a detailed analysis of a simplified case. 

\bibliographystyle{IEEEbib-abbrev}
\bibliography{ZoSGA.bib}
\end{document}